%% file: main.tex
\documentclass[11pt]{article}

\usepackage{amsmath,amssymb,amsthm,amsfonts,latexsym,bbm,xspace,graphicx,float,mathtools,mathdots,cancel,mathtools}
\usepackage{multirow}

\usepackage[backref,bookmarks=true]{hyperref}

\renewcommand*{\backref}[1]{}
\renewcommand*{\backrefalt}[4]{{\footnotesize [%
    \ifcase #1 Not cited.%
	\or Cited on Section~#2%
	\else Cited on Sections~#2%
	\fi%
]}}

\usepackage[dvipsnames]{xcolor}
\hypersetup{
    colorlinks=true,
    linkcolor=Blue,
    citecolor=Mahogany,
    urlcolor=Mahogany,
}

\usepackage{orcidlink}

\usepackage[nameinlink]{cleveref}
\usepackage{thm-restate}
\usepackage[letterpaper,margin=1in]{geometry}
\usepackage{microtype}
\usepackage{tikz}
\usetikzlibrary{shapes.geometric,  arrows.meta, positioning,  matrix}
\usepackage{mathabx}
\usepackage{tcolorbox}

\usepackage{titling}
\usepackage[affil-it]{authblk}

\usepackage{stackengine}
\usepackage{subcaption}
\usepackage{pgfplots}
\usepackage{pgffor} 
\usepackage{booktabs}
\usepackage[round]{natbib}

\newtheorem{theorem}{Theorem}
\newtheorem{proposition}[theorem]{Proposition}
\newtheorem{lemma}[theorem]{Lemma}

\newtheorem{definition}[theorem]{Definition}
\newtheorem{example}[theorem]{Example}

\usepackage{newtxtext, newtxmath}

\usepackage[ruled,vlined,linesnumbered,algonl]{algorithm2e}
\SetEndCharOfAlgoLine{}
\SetKwComment{Comment}{\footnotesize$\triangleright$\ }{}

\SetCommentSty{mycommfont}

\definecolor{codegreen}{rgb}{0,0.6,0}
\definecolor{codegray}{rgb}{0.5,0.5,0.5}
\definecolor{codepurple}{rgb}{0.58,0,0.82}
\definecolor{backcolour}{rgb}{0.95,0.95,0.92}

\usepackage{listings}
\newcommand{\listingsttfamily}{\fontfamily{IBMPlexMono-TLF}\scriptsize}
\lstdefinestyle{mystyle}{
    backgroundcolor=\color{backcolour},   
    commentstyle=\color{codegreen},
    keywordstyle=\color{magenta},
    numberstyle=\tiny\color{codegray},
    stringstyle=\color{codepurple},
    breakatwhitespace=false,         
    breaklines=true,                 
    captionpos=b,                    
    keepspaces=true,                 
    numbers=left,                    
    numbersep=5pt,                  
    showspaces=false,                
    showstringspaces=false,
    showtabs=false,                  
    tabsize=2,
	basicstyle=\listingsttfamily
}
\pgfplotsset{compat=1.18}
\usepackage{xcolor}
\usetikzlibrary{calc}

\stackMath

\newcommand{\squig}{{\scriptstyle\sim\mkern-3.9mu}}

\newcommand{\rsquigend}{{\scriptstyle\rule{.1ex}{0ex}\rhd}}
\newcounter{sqindex}
\newcommand\squigs[1]{%
  \setcounter{sqindex}{0}%
  \whiledo {\value{sqindex}< #1}{\addtocounter{sqindex}{1}\squig}%
}
\newcommand\rsquigarrow[2]{%
  \mathbin{\stackon[2pt]{\squigs{#2}\rsquigend}{\scriptscriptstyle\text{#1\,}}}%
}

\definecolor[named]{ACMBlue}{cmyk}{1,0.1,0,0.1}
\definecolor[named]{ACMYellow}{cmyk}{0,0.16,1,0}
\definecolor[named]{ACMOrange}{cmyk}{0,0.42,1,0.01}
\definecolor[named]{ACMRed}{cmyk}{0,0.90,0.86,0}   
\definecolor[named]{ACMLightRed}{cmyk}{0,0,0,0.35}
\definecolor[named]{ACMLightBlue}{cmyk}{0.49,0.01,0,0}
\definecolor[named]{ACMGreen}{cmyk}{0.20,0,1,0.19}
\definecolor[named]{ACMPurple}{cmyk}{0.55,0.6,0.1,0.15}
\definecolor[named]{ACMPurple2}{cmyk}{0.04,0.7,0.01,0.01}
\definecolor[named]{ACMPurple3}{cmyk}{0.04,0.7,0.01,0.01}
\definecolor[named]{ACMDarkBlue}{cmyk}{1,0.58,0,0.21}

\definecolor{redorange}{rgb}{0.878431, 0.235294, 0.192157}
\definecolor{lightblue}{rgb}{0.552941, 0.72549, 0.792157}
\definecolor{clearyellow}{rgb}{0.964706, 0.745098, 0}
\definecolor{midyellow}{rgb}{0.764706, 0.645098, 0.5}
\definecolor{clearorange}{rgb}{0.917647, 0.462745, 0}
\definecolor{mildgray}{rgb}{0.54902, 0.509804, 0.47451}
\definecolor{softblue}{rgb}{0.643137, 0.858824, 0.909804}
\definecolor{bluegray}{rgb}{0.141176, 0.313725, 0.603922}
\definecolor{lightgreen}{rgb}{0.709804, 0.741176, 0}
\definecolor{redpurple}{rgb}{0.835294, 0, 0.196078}
\definecolor{midblue}{rgb}{0, 0.592157, 0.662745}
\definecolor{clearpurple}{rgb}{0.67451, 0.0784314, 0.352941}
\definecolor{browngreen}{rgb}{0.333333, 0.313725, 0.145098}
\definecolor{darkestpurple}{rgb}{0.396078, 0.113725, 0.196078}
\definecolor{greypurple}{rgb}{0.294118, 0.219608, 0.298039}
\definecolor{darkturqoise}{rgb}{0, 0.239216, 0.298039}
\definecolor{darkbrown}{rgb}{0.305882, 0.211765, 0.160784}
\definecolor{midgreen}{rgb}{0.560784, 0.6, 0.243137}
\definecolor{darkred}{rgb}{0.576471, 0.152941, 0.172549}
\definecolor{darkpurple}{rgb}{0.313725, 0.027451, 0.470588}
\definecolor{darkestblue}{rgb}{0, 0.156863, 0.333333}
\definecolor{lightpurple}{rgb}{0.776471, 0.690196, 0.737255}
\definecolor{softgreen}{rgb}{0.733333, 0.772549, 0.572549}
\definecolor{offwhite}{rgb}{0.839216, 0.823529, 0.768627}
\definecolor{brightgreen}{rgb}{0.85, 0.98, 0.01}

\newcommand{\cola}{cyan} 
\newcommand{\colc}{red!70!white} 

\newtcolorbox{mybox}[1]{%
    colback=teal!10,
    coltitle=black,
    colframe=teal!30,
    fonttitle=\bfseries,
    title=#1, 
    sharp corners,
    boxrule=0pt,
    enhanced,
    overlay={\node[font=\Huge, text=cyan!70!black] at ([yshift=-4mm]interior.north west) {\ding{228}};}
    }

\DeclareMathAlphabet{\mathmybb}{U}{bbold}{m}{n}

\renewcommand{\O}{\mathcal{O}}

\title{{\bfseries\huge{Asymptotically Smaller Encodings\\ for Graph Problems and Scheduling}}} 

\author{Bernardo Subercaseaux~\orcidlink{0000-0001-2345-6789} (\url{bersub@cmu.edu})\thanks{This research is supported by the NSF under grant DMS-2434625.}}
\affil{Carnegie Mellon University, Pittsburgh PA}

\begin{document}

\maketitle

\begin{abstract}
    We show how several graph problems (e.g., vertex-cover, independent-set, $k$-coloring) can be encoded into CNF using only $\O(|V|^2 / \lg |V|)$ many clauses, as opposed to the $\Omega(|V|^2)$ constraints used by standard encodings. This somewhat surprising result is a simple consequence of a result of~\citet{chungDecompositionGraphsComplete1983} about biclique coverings of graphs, and opens theoretical avenues to understand the success of \emph{Bounded Variable Addition}~\citep{mantheyAutomatedReencodingBoolean2012} as a preprocessing tool. Finally, we show a novel encoding for independent sets in some dense interval graphs using only $\O(|V| \lg |V|)$ clauses (the direct encoding uses $\Omega(|V|^2)$), 
    which we have successfully applied to a string-compression encoding posed by~\citet{bannai_et_al:LIPIcs.ESA.2022.12}. As a direct byproduct, we obtain a reduction in the encoding size of a scheduling problem posed by~\citet{mayankEfficientSATEncoding2020} from $\O(N\!MT^2)$ to $\O(N\!MT + M T^2 \lg T)$, where $N$ is the number of tasks, $T$ the total timespan, and $M$ the number of machines.
\end{abstract}

\section{Introduction}
Using a more compact CNF encoding can make the entire difference between a combinatorial problem being solvable (even in many CPU years) and it being intractable~\citep{subercaseauxPackingChromaticNumber2023d,heuleHappyEndingEmpty2024c,wesleyLowerBoundsBook2024,heuleSATApproachCliqueWidth2015,schidlerSATbasedDecisionTree2024,qian2025unfoldingboxeslocalconstraints}.
However, besides a few very general principles~\citep{DBLP:journals/jsat/Bjork11,Prestwich2021}, it seems that the \emph{``art of encodings''} is still mostly explored through problem-specific ideas, and it is not clear how to systematically obtain smaller encodings for combinatorial problems. Furthermore, lower bounds on the size of encodings have been elusive, with very few exceptions on relatively simple constraints such as \emph{Parity}~\citep{emdinCNFEncodingsParity2022} or \emph{At-Most-One}~\citep{kuceraLowerBoundCNF2019}, making it difficult to predict whether the direct encoding for a given problem is already optimal or not. 

In this article, I will show that several standard graph problems can be encoded more efficiently than through their direct formulation, and more importantly, that the tools used can shed light into theoretical questions about encodings.

As a representative example, consider first the \emph{independent set} problem.
The input is a graph $G = (V, E)$, together with an integer $k$, and the goal is to find a subset of the vertices $S \subseteq V$ such that  $\binom{S}{2} \cap E = \varnothing$ (i.e., no two vertices in $S$ are neighbors) and $|S| = k$. The \emph{``direct encoding''} is thus to create, for each vertex $v \in V$, a variable $x_v$ representing whether $v \in S$. Then, the direct encoding consists of enforcing the independent-set property
\begin{equation}\label{eq:binaries}
    \bigwedge_{\{u, v\} \in E} (\overline{x_{u}} \lor \overline{x_{v}}), 
\end{equation}
and then the cardinality constraint
\(   
\sum_{v \in V} x_v = k.\)
While cardinality constraints are known to admit compact encodings with $\O(n)$ clauses~\citep{sinzOptimalCNFEncoding2005a}, and even arc-consistency with $\O(n \lg ^2 n)$ clauses~\citep{asinCardinalityNetworksTheoretical2011}, 
\Cref{eq:binaries} amounts to $\Theta(|E|)$ clauses, which is $\Omega(|V|^2)$ for dense graphs.

Our first contribution (in~\Cref{sec:biclique_encodings}) is to show that this encoding can be improved to $\O(|V|^2 / \lg |V|)$ clauses, and in consequence, that several structurally similar graph problems can be encoded compactly as well:

\begin{theorem}[Informal]\label{thm:main}
    The \emph{independent set}, \emph{vertex cover}, \emph{$k$-coloring}, and \emph{clique} problems can be encoded into CNF using $\O(|V|^2 / \lg |V|)$ many clauses.
\end{theorem}

This result improves upon an idea of~\citet{Rintanen}, and then~\citet{ignatievCardinalityEncodingsGraph2017}, who used \emph{clique coverings} to encode the independent-set property by observing that, for any clique $K_t$ of $G$, at most one vertex of the clique can be part of $S$, which can be encoded using $\O(t)$ clauses as opposed to the $\Omega(t^2)$ clauses used by~\Cref{eq:binaries}. 
However, as the authors themselves note, this idea is not enough to obtain an encoding with $o(|V|^2)$ clauses in all graphs, since for example a complete bipartite graph has $\Omega(|V|^2)$ edges and yet not cliques of size larger than $2$.
We overcome this limitation by using \emph{biclique coverings} of graphs, leveraging the fact that any graph with $\Omega(|V|^2)$ edges must contain a biclique $K_{t, t}$ with $t = \Omega(\lg |V|)$~\citep{chungDecompositionGraphsComplete1983}.

Then, in~\Cref{sec:frameworks}, we compare more in detail the~\emph{biclique covering} framework with the~\emph{clique covering} framework of~\citet{ignatievCardinalityEncodingsGraph2017} as well as with \emph{Bounded Variable Addition} (BVA)~\citep{mantheyAutomatedReencodingBoolean2012}, a successful preprocessing technique for reducing encoding sizes.
As a cornerstone of this comparison, we study in~\Cref{sec:disjoint-intervals} how to encode that a selection of intervals $[i,  j]$ for $1 \leq i < j \leq n$ is \emph{pairwise disjoint}, i.e., that no two intervals overlap. This corresponds to encoding the independent-set property of a corresponding interval graph $\mathcal{I}_n$. We show that, despite the fact that $|E(\mathcal{I}_n)| = \Omega(n^4)$, it is possible to obtain a much more compact encoding.

\begin{theorem}[Informal]\label{thm:intervals}
    The independent-set property of the interval graph $\mathcal{I}_n$ can be encoded into CNF using $\O(n^2 \lg n)$ clauses.
\end{theorem}

We show that this is more efficient than what is obtained via either the clique covering framework or the biclique covering framework. Moreover, we show that while BVA can obtain a more compact encoding in terms of the number of clauses (experimentally, since we do not have a theoretical guarantee), the structured encoding we present works better in practice. This is reminiscent of the idea of~\citet{haberlandtEffectiveAuxiliaryVariables2023} for \emph{Structured Bounded Variable Addition} (SBVA), which ended up winning the SAT competition 2023.

\subsection{Notation and preliminaries}

We will use $\top$ and $\bot$ to denote \emph{true} and \emph{false}, respectively. The negation of a variable $x$ is denoted $\overline{x}$, and a \emph{literal} is either a variable or the negation of a variable. Literals $x$ and $\overline{x}$ are said to be complementary, for any variable $x$. A \emph{clause} is a set of non-complementary literals, and a \emph{formula} is a set of clauses (we will thus identify $\land$ with $\cup$ when operating over clauses, and $\ell_1 \lor \ell_2$ with $\{\ell_1\} \cup \{\ell_2\}$ when operating over literals). The size of a formula is simply its number of clauses.
We denote the set of variables appearing in a formula $F$ as $\textsf{Var}(F)$. 
Given a set $\mathcal{V}$ of variables, an assignment is a function $\tau : \mathcal{V} \to \{\bot, \top\}$. For a variable $x \in \mathcal{V}$, we say $\tau \models x$ if $\tau(x) = \top$, and similarly, $\tau \models \overline{x}$ if $\tau(x) = \bot$. For a clause $C$, we say $\tau \models C$ if $\bigvee{\ell \in C} \tau \models C$, and a for a formula $F$, $\tau \models F$ if $\bigwedge_{C \in F} \tau \models C$. When writing this, we assume implicitly that $\textsf{Var}(F) \subseteq \mathcal{V}$.
For an assignment $\tau: \mathcal{V} \to \{\bot, \top\}$ and a formula $F$, now with $\mathcal{V} \subsetneq \textsf{Var}(F)$, we denote by $F|_\tau$ the formula obtained by eliminating from $F$ each clause satisfied by $\tau$, and then from each remaining clause eliminating every literal $\ell$ such that $\tau \models \overline{\ell}$. Note that $\textsf{Var}(F|_\tau) = \textsf{Var}(F) \setminus \mathcal{V}$.
We will write $\textsf{SAT}(F)$ to say that $\tau \models F$ for some assignment $\tau$, and $\textsf{UNSAT}(F)$ to mean that no such assignment exists.

\section{Subquadratic Encodings through Biclique Coverings}\label{sec:biclique_encodings}

\subsection{Clique Covering Encodings}
\newcommand{\amope}{\textsf{AMO}_{\textsf{PE}}}
Arguably, the \emph{At-Most-One} constraint (\textsf{AMO}) for variables $x_1, \ldots, x_n$ is the most elemental example of an encoding whose direct formulation can be asymptotically improved. The naïve formulation, often referred to as the \emph{pairwise encoding}, is simply
\[
\textsf{AMO}(x_1, \ldots, x_n) := \bigwedge_{1 \leq i \neq j \leq n} (\overline{x_i} \lor \overline{x_j}),
\]
using $\binom{n}{2}$ clauses, analogously to the dense case of~\Cref{eq:binaries}. On the other hand, several formulations using $\O(n)$ clauses are known~\citep{empirical,zhou2020comparisonsatencodingsatmostk}. The most compact is Chen's \emph{product encoding}, which uses $2n + 4\sqrt{n} + \O(\sqrt[4]{n})$ clauses~\citep{chen2010new}, and is essentially tight~\citep{kuceraLowerBoundCNF2019}. We will use notation $\textsf{AMO}_{\textsf{PE}}(x_1, \ldots, x_n)$ to denote the formula resulting from Chen's product encoding over variables $x_1, \ldots, x_n$.\footnote{We will rather arbitrarily use Chen's product encoding throughout this section, but any encoding using a linear number of clauses would work as well.}
As noted by~\citet{ignatievCardinalityEncodingsGraph2017}, we can interpret this result as saying that \emph{``the independent-set property can be encoded in $\O(n)$ clauses for a complete graph $K_n$''}. Let us now formalize what this means. 

\begin{definition}[Encoding the ISP]
     \label{def:isp_encoding}
    Given a graph $G = (V, E)$,  a set of variables $X = \{ x_v \mid v \in V\}$, and a potentially empty set of (auxiliary) variables $Y = \{y_1, \ldots, y_m\}$, we say that a formula $F$ with $\textsf{var}(F) = X \cup Y$ encodes the ``independent-set property'' (ISP) if for every assignment $\tau : X \to \{\bot, \top\}$,
    \[
        \textsf{SAT}(F |_\tau) \iff \{ v \in V \mid \tau(x_v) = \top \} \text{ is an independent set of } G.
    \]
\end{definition}
The trivial observation now is that for a complete graph $K_n$, a subset $S \subseteq V(K_n)$ of its vertices is independent if and only if $|S| \leq 1$, and thus $\textsf{AMO}_{\textsf{PE}}(V(K_n))$ encodes the ISP of $K_n$ with $\O(n)$ clauses. Note that we have written $\textsf{AMO}_{\textsf{PE}}(V(K_n))$, identifying each vertex $v$ with a corresponding variable $x_v$, and will do this repeatedly to avoid cluttering the notation.

To extend this idea to more general graphs,~\citet{ignatievCardinalityEncodingsGraph2017} used~\emph{clique coverings}: a \emph{``clique covering''} of a graph $G$ is a set $\{C_1, \ldots, C_m\}$ of subgraphs of $G$, each of which must be a clique, such that every edge $e \in E(G)$ belongs to some subgraph $C_i$. That is, $\bigcup_{i=1}^m E(C_i) = E(G)$.
We thus define a clique-covering-ISP (CC-ISP) encoding as follows:
\begin{definition}[CC-ISP encoding]
    Given a graph $G$, and a clique covering $\mathcal{C}$ of $G$, the formula 
    \[
        F_\mathcal{C} := \bigwedge_{C \in \mathcal{C}} \amope(V(C))
    \]
    is said to be a ``CC-ISP encoding'' for $G$.  
\end{definition}
Note that the number of clauses of $F_\mathcal{C}$ is 
\(
|F_\mathcal{C}| = \sum_{C \in \mathcal{C}} f(|V(C)|),
\)
where $f(n)$ is the number of clauses used by $\amope$ over $n$ variables, and thus $f(n) = \O(n)$.
\begin{lemma}[Implicit in~\citet{ignatievCardinalityEncodingsGraph2017}]\label{lemma:clique_encoding}
    Any CC-ISP encoding $F_\mathcal{C}$ for a graph $G$ indeed encodes the independent-set property (see~\Cref{def:isp_encoding}) for $G$.
\end{lemma}
\begin{proof}
    Let $\tau$ be an assignment of the variables $\{x_v \mid v \in V(G)\}$, and $S_\tau = \{v \in V(G) \mid \tau(x_v) = \top\}$.
    Now, assume $S_\tau$ is an independent set of $G$. As the intersection of any independent set and a clique has at most one vertex, we have $|S_\tau \cap V(C)| \leq 1$ for every $C \in \mathcal{C}$, and thus
    $F_\mathcal{C}|_\tau$ is clearly satisfiable. Conversely, if $S_\tau$ is not an independent set of $G$, then some edge $e = \{u, v\}$ of $G$ has $|S_\tau \cap e| = 2$, but by definition of clique covering, $e \in E(C)$ for some $C \in \mathcal{C}$, and thus $e \subseteq V(C)$. But this implies
    $|S_\tau \cap V(C)| \geq 2$, and thus the subformula $\amope(V(C))|_\tau$ is already unsatisfiable, implying $F_\mathcal{C}|_\tau$ is unsatisfiable.
    %
\end{proof}

Unfortunately, this re-encoding technique is not useful for worst-case graphs. For example, in a complete bipartite graph (or \emph{biclique}) $K_{n, n}$, then the maximum clique has size $2$ (otherwise there would be a triangle, contradicting bipartiteness), and thus any clique covering $\mathcal{C}$ of $K_{n, n}$ covers at most one edge per clique, leading to $|\mathcal{C}| \geq |E(K_{n,n})| = n^2$, and thus $|F_\mathcal{C}| \geq n^2$. Recall that $|E(K_{n, n})| = n^2$ is the number of clauses used by the pairwise encoding (\Cref{eq:binaries}), and thus no improvement is achieved.
It is worth noting that clique coverings do yield an improvement for random graphs $G(n, p)$, using a result by~\citet{friezeCoveringEdgesRandom1995} (errors corrected in its arXiv version~\citep{frieze2011coveringedgesrandomgraph}).

\begin{proposition}\label{prop:random_coverings}
    Let $G \sim G(n, p)$ be a random graph on $n$ vertices obtained by adding each edge independently with probability $p$. Then, with high probability there exists a clique covering $\mathcal{C}$ for $G$ such that 
    \(
        |F_\mathcal{C}| = \O\left(\frac{n^2}{\lg n}\right).
    \)
\end{proposition}
\begin{proof}[Proof Sketch]
    The covering $\mathcal{C}$ of~\citet{frieze2011coveringedgesrandomgraph} is constructed iteratively, where the $i$-th iteration uses roughly
    \(
        p n^2  i^2   e^{1-i}  (\ln n)^{-2}
    \)
    cliques of size at most $\ln n /i$, for $1 \leq i \leq \lceil 4 \ln \ln n\rceil$. After that,  potentially $\O(n^2 (\lg n)^{-4})$ cliques of size $2$ are added to deal with any uncovered edges, which can safely ignore since this is already $\O(n^2 \lg n)$.
    Thus, we have that 
    \[
        |F_\mathcal{C}| = \O\left(\sum_{C \in \mathcal{C}} |V(C)|\right) = \O\left(n^2 \cdot p \cdot e  \sum_{i=1}^{\lceil 4 \ln \ln n\rceil} \frac{i^2}{(\ln n)^2 e^i} \cdot \frac{\ln n}{i} \right) = \O\left(\frac{n^2}{\ln n} \sum_{i=1}^{\lceil 4 \ln \ln n\rceil} \frac{i}{e^i}\right).
    \]
    Using the standard power series $\sum_{i=1}^\infty i \cdot x^i = \frac{x}{(1-x)^2}$ for $|x| < 1$,   we have 
    \[
        \sum_{i=1}^{\lceil 4 \ln \ln n\rceil} \frac{i}{e^i} < \sum_{i=1}^{\infty} i (1/e)^i = \frac{e^{-1}}{(1 - 1/e)^2} = \O(1),
    \]
    so $|F_\mathcal{C}| = \O(n^2 / \lg n)$ as desired.
\end{proof}

We will next see that by using biclique coverings instead of clique coverings we get a worst-case asymptotic improvement.

\subsection{Biclique Covering Encodings}
For a biclique $K_{a, b}$ (complete bipartite graph with $a$ vertices on one part and $b$ on the other part), the independent-set property can also be encoded efficiently. Let $A$ and $B$ be the parts of $K_{a, b}$. Then, introduce an auxiliary variable $x_A$  which intuitively represents that some vertex of $A$ belongs to the desired independent set $S$. The desired formula is then
\begin{equation}\label{eq:biclique_encoding}
\textsf{BIS}(K_{a, b}) := \left(\bigwedge_{v \in A}(\overline{x_v} \lor x_A)\right) \land \left(\bigwedge_{v \in B}(\overline{x_A} \lor \overline{x_v}) \right).
\end{equation}
Intuitively, the first part of the formula is saying that if some vertex $v \in A$ is selected, then $x_A$ will be true, and the second part enforces that $x_A$ being true forbids any vertex in $B$ from being selected.
Note immediately that $|\textsf{BIS}(K_{a, b})| = a + b = |V(K_{a, b})|$, as opposed to the direct encoding which uses $|E(K_{a, b})| = a \cdot b$. As can be observed in~\Cref{fig:biclique_encoding} (dashed purple box), this wastes a clause when $a = b = 1$, so we shall assume that in this case $\textsf{BIS}(K_{1, 1})$ will be the direct encoding. 
We again lift the biclique encoding to arbitrary graphs by coverings. A \emph{biclique covering} $\mathcal{B}$ for a graph $G$ is simply a set of bicliques $B_1, \ldots, B_m$ that are subgraphs of $G$ and such that $\bigcup_{i=1}^m E(B_i) = E(G)$.
\begin{proposition}\label{prop:biclique-independence}
    Let $\mathcal{B} = \{B_1, \ldots, B_m\}$ be a biclique covering of a graph $G$, and $S \subseteq V(G)$ some set of vertices. Then, $S$ is an independent set of $G$ if and only if for every $1 \leq i \leq m$, $S \cap V(B_i)$ is an independent set of $B_i$.
\end{proposition}
\begin{proof}
    Suppose first that $S$ is an independent set of $G$. Then trivially $S \cap V(B_i)$ is also an independent set of $G$ for every $1 \leq i \leq m$, and as each $B_i$ is a subgraph of $G$, the sets $S \cap V(B_i)$ are also independent sets of $B_i$.
    For the opposite direction, assume that $S$ is not an independent set of $G$, and thus $|e \cap S| = 2$ for some $e \in E(G)$. Then, as $e \in E(B_i)$ for some $i \in \{1, \ldots, m\}$ by definition of covering, we have $e \cap V(B_i) = e$ and thus
    \[
    |e \cap (S \cap V(B_i))| = |(e \cap V(B_i)) \cap S| = |e \cap S| = 2,
    \]
    implying $S \cap V(B_i)$ is not independent in $B_i$.
\end{proof}
Using~\Cref{prop:biclique-independence} we directly obtain the following analog to~\Cref{def:isp_encoding} and~\Cref{lemma:clique_encoding}:
\begin{proposition}\label{prop:biclique_encoding}
     Given a graph $G$, and a biclique covering $\mathcal{B}$ of $G$, the formula 
    \(
        F_\mathcal{B} := \bigwedge_{B \in \mathcal{B}} \textsf{BIS}(B)
    \)
    is said to be a ``BC-ISP'' encoding for $G$. The formula $F_\mathcal{B}$
     encodes the independent-set property of $G$, and has size 
    $\sum_{B \in \mathcal{B}} |V(B)|$.
\end{proposition}

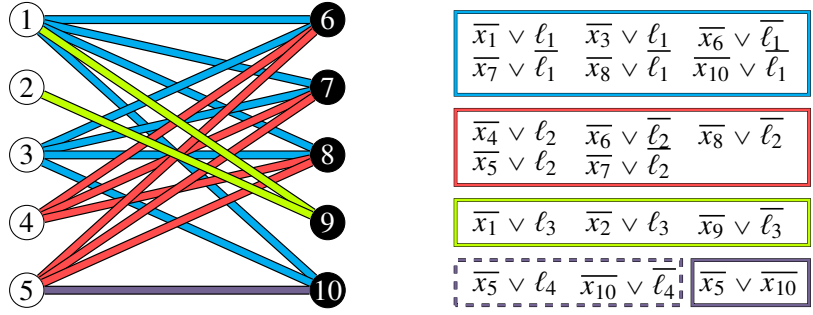
\begin{figure}
\centering
\input{biclique_decomp.tex}
\caption{A biclique covering of a bipartite graph with $10$ vertices and $17$ edges, resulting in a formula with 15 clauses. }\label{fig:biclique_encoding}
\end{figure}

The key difference now is made by a result stating that every graph admits a biclique covering that is asymptotically smaller than just taking its set of edges.
\begin{theorem}[\citet{chungDecompositionGraphsComplete1983}]\label{thm:erdos_chung_spencer}
    Every graph $G$ on $n$ vertices has a biclique covering $\mathcal{B}$ with 
    \(
        \sum_{B \in \mathcal{B}} |V(B)| = \O(n^2 / \lg n).
        \)
\end{theorem}

Combining~\Cref{thm:erdos_chung_spencer} and~\Cref{prop:biclique_encoding} we get the main result of this section.
\begin{theorem}[Formal version of~\Cref{thm:main}]\label{thm:subquadratic_encoding}
    For every graph $G$ on $n$ vertices, there is a formula $F$ that encodes the independent-set property of $G$ such that $|F| = \O(n^2 / \lg n)$.
\end{theorem}
\Cref{fig:biclique_encoding} illustrates a biclique covering encoding, showing a reduction in size from the direct encoding. By avoiding the re-encoding of $K_{1, 1}$ cliques (lime green and purple in~\Cref{fig:biclique_encoding}), the size would go down to $13$.

We can trivially extend~\Cref{thm:subquadratic_encoding} to other graph problems as we show next.
\begin{itemize}
    \item (\textbf{Vertex Cover}) Since a set of vertices $S \subseteq V(G)$ is a vertex cover if and only if $V(G) \setminus S$ is an independent set, it suffices to invert the polarity of each literal $x_v$ (resp. $\overline{x_v}$) to $\overline{x_v}$ (resp. $x_v$) in the formula $F$ from~\Cref{thm:subquadratic_encoding}, resulting on a formula of the same size.
    \item ($k$-\textbf{Coloring}) A $k$-coloring of a graph $G$ is the same as a partition of $V(G)$ into $k$ independent sets. Naturally, if we have variables $x_{v, c}$ for $v \in V(G)$ and $c \in \{1, \ldots, k\}$ that indicate assigning color $c$ to vertex $v$, then we simply use the conjunction of $k$ formulas obtained from~\Cref{thm:subquadratic_encoding}, where the $c$-th of them is obtained by replacing each variable $x_v$ by $x_{v, c}$.
    \item (\textbf{Clique}) It suffices to use that a set of vertices $S \subseteq V(G)$ is a clique of $G$ if and only if $S$ is an independent set of the complement graph $\overline{G}$. 
\end{itemize}

We conclude this section by noting that~\Cref{thm:subquadratic_encoding}, and its analog to the other graph problems mentioned above can be made constructive by using a result of~\citet{mubayiFindingBipartiteSubgraphs2010}, which states that a biclique covering with the asymptotic guarantee of~\Cref{thm:erdos_chung_spencer} can be computed deterministically in polynomial time.

\section{Covering Frameworks}\label{sec:frameworks}

Even though the presented biclique-covering encodings are more efficient than clique-covering encodings over worst-case graphs, this is not necessarily the case on every graph. The main example being again $K_n$, for which we have the trivial clique covering $\mathcal{C} = \{K_n\}$, but for which
every biclique covering uses at least $\lfloor \lg n \rfloor$ many bicliques~\cite{fishburnBipartiteDimensionsBipartite1996}. Furthermore, the trivial clique covering results in formula with $\O(n)$ clauses, whereas any biclique covering results in $\Omega(n \lg n)$ clauses.
\begin{proposition}
    There is a BC-ISP encoding for $K_n$ using $\O(n \lg n)$ clauses, and any BC-ISP encoding for $K_n$ uses $\Omega(n \lg n)$ many clauses.
\end{proposition}
\begin{proof}
For the upper bound, we construct a biclique covering $\mathcal{B}$ recursively. First, we separate $V(K_n)$ into two parts $L$ and $R$, with $|L| = \lceil n / 2\rceil$ and $|R| = \lfloor n/2 \rfloor$. Then, we add to $\mathcal{B}$ the biclique $K(L, R)$ between the vertices in $L$ and the vertices in $R$, and proceed recursively on both $L$ and $R$ obtaining biclique coverings $\mathcal{B}_L$ and $\mathcal{B}_R$ respectively. The total biclique covering of $K_n$ is then given by 
\(
\mathcal{B} = K(L, R) \cup \mathcal{B}_L \cup \mathcal{B}_R.
\)
Thus, if we let $g(n)$
denote the number of clauses used in an optimal BC-ISP encoding for $K_n$ (i.e., one that minimizes $\sum_{B \in \mathcal{B}} |V(B)|$), we have 
\(
    g(n) \leq n + g(\lfloor n/2 \rfloor) + g(\lceil n/2 \rceil) \leq n + 2 g(\lceil n/2\rceil),
\)
from where it follows that $g(n) = \O(n \lg n)$.

The lower bound follows from a nice result of~\citet{dongDecompositionGraphsComplete2007}, which states that 
in any biclique covering $\mathcal{B}$ of $K_n$, every vertex $v \in V(K_n)$ appears in at least $\lfloor \lg n \rfloor$ many bicliques in $\mathcal{B}$. Indeed, using that result it suffices to say that 
\begin{align*}
|F_\mathcal{B}| &= \sum_{B \in \mathcal{B}} |V(B)| = \sum_{B \in \mathcal{B}} \sum_{v \in V(K_n)} \mathmybb{1}_{v \in B}
 = \sum_{v \in V(K_n)} \sum_{B \in \mathcal{B}} \mathmybb{1}_{v \in B}  \geq \sum_{v \in V(K_n)} \lfloor \lg n \rfloor = \Omega(n \lg n). \qedhere
\end{align*}
\end{proof}

\subsection{Bounded Variable Addition}

Interestingly, it is possible to do better for $K_n$ by using a slight generalization of biclique coverings, as the~\emph{Bounded Variable Addition} (BVA)~\citep{mantheyAutomatedReencodingBoolean2012} preprocessing technique does. Let us explain first how BVA operates, which requires the following definition adapted from~\citet{haberlandtEffectiveAuxiliaryVariables2023}.
\begin{definition}[Grid product]
    Given a clause $L$, and set of clauses $\Gamma$, the \emph{``grid product''} of $L$ and $\Gamma$ is the set of clauses
    \(
    L \bowtie \Gamma := \bigcup_{\ell \in L} \bigcup_{\gamma \in \Gamma} \gamma \cup \{\ell\}. 
    \)
\end{definition}

BVA works by iteratively identifying subsets of a formula $F$ that can be expressed as a grid product $L \bowtie \Gamma$ (note that this does not require $L \in F$ nor $\Gamma \subseteq F$.), and introducing a new auxiliary variable $y$ which allows replacing $L \bowtie \Gamma$ by the clauses
\(
    \bigwedge_{\ell \in L} (\overline{y} \lor \ell) \land \bigwedge_{\gamma \in \Gamma} (y \lor \gamma).
\)
Naturally, resolving on the new variable $y$ yields the replaced set of clauses, and thus the new formula is equisatisfiable to the original one.
\begin{example}
    Let $L = (x_1 \lor x_2)$ and $\Gamma = \{(p \lor q), (q \lor r), (\overline{p} \lor \overline{r} \lor \overline{q})\}$, then the corresponding grid product is
    \[
    L \bowtie \Gamma = \{(x_1 \lor p \lor q), (x_1 \lor q \lor r), (x_1 \lor \overline{p} \lor \overline{r} \lor \overline{q}), (x_2 \lor p \lor q), (x_2 \lor q \lor r), (x_2 \lor \overline{p} \lor \overline{r} \lor \overline{q})\}.
    \]
    The new variable $y$ can be used to replace this set of clauses by the following ones:
    \[
    (\overline{y} \lor x_1) \land (\overline{y} \lor x_2) \land (y \lor p \lor q) \land (y \lor q \lor r) \land (y \lor \overline{p} \lor \overline{r} \lor \overline{q}).
    \]
\end{example}
\begin{definition}
    Given a formula $F$, and a subset of $F$ that can be expressed as a grid product $L \bowtie \Gamma$, we say that the formula
    \(
    F' := F \setminus (L \bowtie \Gamma) \land \bigwedge_{\ell \in L} (\overline{y} \lor \ell) \land \bigwedge_{\gamma \in \Gamma} (y \lor \gamma)
    \)
    is obtainable from $F$ by BVA, which we denote by $F \xrightarrow{\textsf{BVA}} F'$. If there is a sequence of formulas $F_1, \ldots, F_k$ such that
    \(
        F \xrightarrow{\textsf{BVA}} F_1 \xrightarrow{\textsf{BVA}} \ldots \xrightarrow{\textsf{BVA}} F_k,
    \)
    we say that $F_k$ is a potential BVA re-encoding of $F$ and write $F \rsquigarrow{\textsf{BVA}}{3} F_k$.
\end{definition}

Note that in the case of $\Gamma$ being a set of unit clauses, this matches~\Cref{eq:biclique_encoding}, from where we immediately have the following result:
\begin{proposition}\label{prop:bva_biclique}
    For every graph $G$ on $n$ vertices, there is a formula $F$ such that 
    \[\bigwedge_{\{u, v\} \in E(G)} (\overline{x_u} \lor \overline{x_v}) \rsquigarrow{\textsf{BVA}}{3} F\] and $|F| = \O(n^2 / \lg n)$.
\end{proposition}
    
An important difference between BVA and covering re-encodings is that the new variables created by BVA can also be part of the identified grid products, whereas the (bi)clique covering encodings only use the original set of variables.  
We will show that this difference is enough to obtain a linear encoding for the independent-set property of $K_n$, and thus showing that BVA re-encodes pairwise cardinality constraints into a linear number of clauses. While this was already observed without proof in the original BVA paper of~\citet{mantheyAutomatedReencodingBoolean2012}, we remark that the authors have declared that their results were only empirical~\citep{githubExactlyClauses}. We now provide a formal proof, noting that it applies to an idealized version of BVA, as opposed to its actual implementation, which is more subtle.

\begin{proposition}
    There is a formula $F$ such that $\textsf{AMO}(x_1, \ldots, x_n) \rsquigarrow{\textsf{BVA}}{3} F$, and $|F| =  \O(n)$.
\end{proposition}
\begin{proof}
    We prove that for $n \geq 3$, there is such an $F$ with $|F| = 3n - 6$. The proof is by induction on $n$. For $n = 3$ and $n=4$, we simply take $F := \textsf{AMO}(x_1, \ldots, x_n)$, which has size $\binom{3}{2} = 3 = 3 \cdot 3 - 6$ and $\binom{4}{2} = 6 = 3 \cdot 4 - 6$ respectively. 
    For $n \geq 5$, consider first the grid product $\{x_1, x_2, x_3\} \bowtie \{x_4, x_5, \ldots, x_n\}$, which is clearly in $\textsf{AMO}(x_1, \ldots, x_n)$. The replacement of this grid product by BVA can be split into the following sets of clauses:
    \begin{enumerate}
        \item The clauses $(\overline{x_1} \lor \overline{x_2})$, $(\overline{x_1} \lor \overline{x_3})$, and $(\overline{x_2} \lor \overline{x_3})$, unaltered by the replacement.
        \item The clauses $(\overline{y} \lor \ell)$ for $\ell \in \{\overline{x_1}, \overline{x_2}, \overline{x_3}\}$.
        \item The clauses $(y \lor \gamma)$ for $\gamma \in \{\overline{x_4}, \overline{x_5}, \ldots, \overline{x_n}\}$.
        \item The clauses $(\overline{x_i} \lor \overline{x_j})$ for $4 \leq i < j \leq n$, also unaltered by the replacement.
    \end{enumerate}
    There are only 6 clauses of types (1) and (2), and for the clauses of types (3) and (4), the key observation is that they correspond exactly to the formula
    \(
        \textsf{AMO}(\overline{y}, x_4, \ldots, x_n),
    \)
    which by the induction hypothesis admits a BVA re-encoding $F'$ of size $3(n-2) - 6$. Thus, denoting by $S$ the set of clauses of types (1) and (2), we have 
    \[
    \textsf{AMO}(x_1, \ldots, x_n) \xrightarrow{\textsf{BVA}} S \cup \textsf{AMO}(\overline{y}, x_4, \ldots, x_n) \rsquigarrow{\textsf{BVA}}{3} S \cup F',
    \]
    and $|S \cup F'| = 6 + 3(n-2) - 6 = 3n - 6$, as desired.
\end{proof}

\subsection{A Lower Bound for the Independent-Set Property}

We now turn our attention to whether BVA, or some other re-encoding technique, could yield an asymptotic improvement over~\Cref{thm:subquadratic_encoding}. While we do not manage to answer this question in full generality, we show that if we restrict ourselves to $2$-CNF formulas (as the ones obtained by BVA, or the covering encodings), then we can prove a matching lower bound of $\Omega(n^2 / \lg n)$, using a similar argument to~\citet[Theorem 1.7]{juknaComputationalComplexityGraphs2013}. In fact our results holds for $k$-CNF (at most $k$ literals per clause) for any fixed $k$.

\begin{proposition}\label{thm:lower_bound}
    Fix an integer $k \geq 2$. Then, for any sufficiently large $n$, there is a graph $G_n$ on $n$ vertices such that any $k$-CNF formula encoding the independent-set property of $G_n$ has size $\Omega(n^2 / \lg n)$.
\end{proposition}
\begin{proof}
    There are $2^{\binom{n}{2}}$ many distinct $n$-vertex graphs, and we will prove first that each of them requires a different formula to encode its independent-set property. Assume, expecting a contradiction, that some formula $F$ encodes the independent-set property for two distinct $n$-vertex graphs $G$ and $G'$. Then, take an edge $\{u, v\}$ that is in $G$ but not in $G'$ (without loss of generality), and consider the assignment $\tau$ defined by
    \(
    \tau(x_w) = \begin{cases}
        \top & \text{if } w \in \{u, v\},\\
        \bot & \text{otherwise.}
    \end{cases}
    \)
    Then, according to~\Cref{def:isp_encoding}, since $F$ encodes the independent-set property for $G$ we have that $F_\tau$ is satisfiable iff $\{u, v\}$ is an independent set in $G$, which is not the case since $\{u, v\} \in E(G)$, thus making $F_\tau$ unsatisfiable. However, since $F$ also encodes the independent-set property for $G'$, we conclude that $F_\tau$ is satisfiable, which is a contradiction.
    Now, we fix without loss of generality the set of variables for the $k$-CNF formulas encoding the independent-set property to be $\{x_1, \ldots, x_n, y_1, \ldots, y_r\}$, where $r$ is the maximum number of auxiliary variables used in any of the formulas.
    Consider now that there are $2^r\binom{m}{r}$ many possible $r$-ary clauses from a set of $m$ variables (since each of the $k$ literals can have two polarities), and thus, the number of $k$-CNF formulas with $t$ clauses is at most 
\[
\binom{\sum_{r=1}^{k} 2^r\binom{m}{r}}{t} \leq \left(k 2^k\binom{m}{k}\right)^t \leq \left(k 2^k \binom{t}{k}\right)^t,
\]
where the last inequality assumes without loss of generality that the formulas do not contain pure literals (which can be simply removed), and thus the number of variables $m$ is at most the number of clauses $t$. As we proved that each different graph needs a different formula, we have that 
\begin{align*}
    & \left(k 2^k \binom{t}{k}\right)^t \geq 2^{\binom{n}{2}} 
    \iff  t \cdot \lg\left(k 2^k\binom{t}{k}\right) \geq \binom{n}{2} \\
    \iff & t \cdot \lg(t) \geq c \cdot n^2 \iff t \geq c \cdot n^2 / \lg(t) \tag{for some constant $c > 0$}.
\end{align*}
But unless $t = \Omega(n^2)$, in which case we can directly conclude the result, we have that $\lg(t) \leq 2 \lg(n)$, and thus we conclude that $t = \Omega(n^2 / \lg n)$.

\end{proof}

\section{Disjoint Intervals}\label{sec:disjoint-intervals}

We now study the independent-set property over a class of graphs that is relevant for a variety of scheduling-related problems: \emph{interval graphs}. 
\begin{definition}
    For any $n \geq 2$, we define the full and discrete interval graph $\mathcal{I}_n$ as the graph whose vertices are all the intervals $[i, j]$ for integers $1 \leq i < j \leq n$, 
    and there are edges between any two intervals that intersect.
\end{definition}

\newcounter{interval}
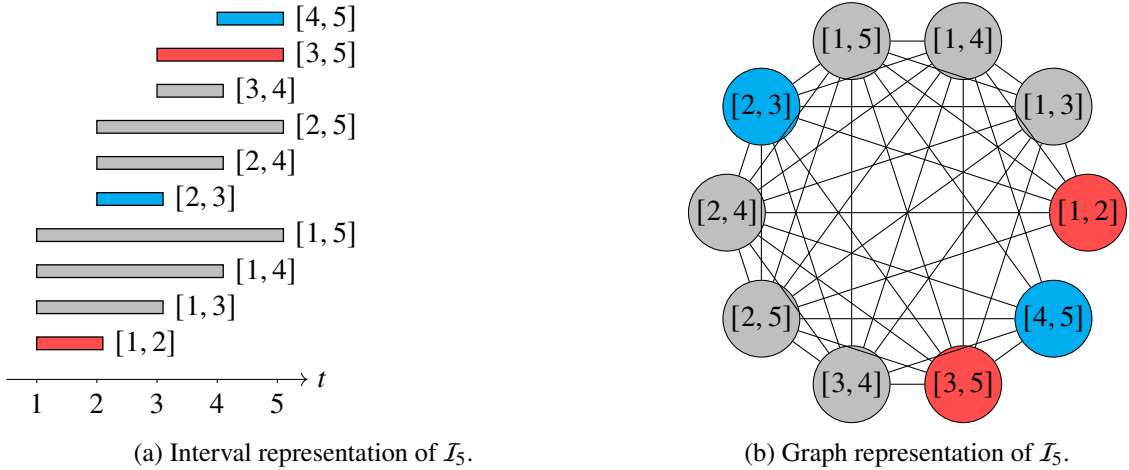
\begin{figure}
    \begin{subfigure}{0.48\textwidth}
        
    \begin{tikzpicture}[x=1cm, y=0.6cm, scale=0.8]
  \def\n{5}

  \setcounter{interval}{0}

        \newcommand{\thickness}{6.0pt}

      \newlength{\thicknesss}
  \setlength{\thicknesss}{6pt}    

  \pgfmathsetlengthmacro{\halfthick}{\thicknesss/2}

  \filldraw[ fill=red!70!white, draw=black, line width=0.5pt]
    ($(1,1)-(0,\halfthick)$) rectangle ($(2.1,1)+(0,\halfthick)$);
  \filldraw[ fill=lightgray, draw=black, line width=0.5pt]
    ($(1,2)-(0,\halfthick)$) rectangle ($(3.1,2)+(0,\halfthick)$);
    \filldraw[ fill=lightgray, draw=black, line width=0.5pt]
        ($(1,3)-(0,\halfthick)$) rectangle ($(4.1,3)+(0,\halfthick)$);
    \filldraw[ fill=lightgray, draw=black, line width=0.5pt]
        ($(1,4)-(0,\halfthick)$) rectangle ($(5.1,4)+(0,\halfthick)$);
    \filldraw[ fill=cyan, draw=black, line width=0.5pt]
        ($(2,5)-(0,\halfthick)$) rectangle ($(3.1,5)+(0,\halfthick)$);
    \filldraw[ fill=lightgray, draw=black, line width=0.5pt]
        ($(2,6)-(0,\halfthick)$) rectangle ($(4.1,6)+(0,\halfthick)$);
    \filldraw[ fill=lightgray, draw=black, line width=0.5pt]
        ($(2,7)-(0,\halfthick)$) rectangle ($(5.1,7)+(0,\halfthick)$);
    \filldraw[ fill=lightgray, draw=black, line width=0.5pt]
        ($(3,8)-(0,\halfthick)$) rectangle ($(4.1,8)+(0,\halfthick)$);
    \filldraw[ fill=red!70!white, draw=black, line width=0.5pt]
        ($(3,9)-(0,\halfthick)$) rectangle ($(5.1,9)+(0,\halfthick)$);
    \filldraw[ fill=cyan, draw=black, line width=0.5pt]
        ($(4,10)-(0,\halfthick)$) rectangle ($(5.1,10)+(0,\halfthick)$);

    \node[black,right] at (2.1,1) {$[1,2]$};
    \node[black,right] at (3.1,2) {$[1,3]$};
    \node[black,right] at (4.1,3) {$[1,4]$};
    \node[black,right] at (5.1,4) {$[1,5]$};
    \node[black,right] at (3.1,5) {$[2,3]$};
    \node[black,right] at (4.1,6) {$[2,4]$};
    \node[black,right] at (5.1,7) {$[2,5]$};
    \node[black,right] at (4.1,8) {$[3,4]$};
    \node[black,right] at (5.1,9) {$[3,5]$};
    \node[black,right] at (5.1,10) {$[4,5]$};

  \draw[->] (0.5,0) -- (\n+0.5,0) node[right] {$t$};
  \foreach \k in {1,...,\n}
    \draw (\k,0) -- (\k,-0.1) node[below] {\(\k\)};
\end{tikzpicture}
\caption{Interval representation of $\mathcal{I}_5$.}\label{fig:intervals}
\end{subfigure}
    \begin{subfigure}{0.48\textwidth}
        \centering
        \begin{tikzpicture}[scale=0.8]
\node[draw, circle, inner sep=1pt, fill=red!70!white] (v1) at (0.00:3cm) {$[1,2]$};
\node[draw, circle, inner sep=1pt, fill=lightgray] (v2) at (36.00:3cm) {$[1,3]$};
\node[draw, circle, inner sep=1pt, fill=lightgray] (v3) at (72.00:3cm) {$[1,4]$};
\node[draw, circle, inner sep=1pt, fill=lightgray] (v4) at (108.00:3cm) {$[1,5]$};
\node[draw, circle, inner sep=1pt, fill=cyan] (v5) at (144.00:3cm) {$[2,3]$};
\node[draw, circle, inner sep=1pt, fill=lightgray] (v6) at (180.00:3cm) {$[2,4]$};
\node[draw, circle, inner sep=1pt, fill=lightgray] (v7) at (216.00:3cm) {$[2,5]$};
\node[draw, circle, inner sep=1pt, fill=lightgray] (v8) at (252.00:3cm) {$[3,4]$};
\node[draw, circle, inner sep=1pt, fill=red!70!white] (v9) at (288.00:3cm) {$[3,5]$};
\node[draw, circle, inner sep=1pt, fill=cyan] (v10) at (324.00:3cm) {$[4,5]$};

\draw (v1) -- (v2);
\draw (v1) -- (v3);
\draw (v1) -- (v4);
\draw (v1) -- (v5);
\draw (v1) -- (v6);
\draw (v1) -- (v7);
\draw (v2) -- (v3);
\draw (v2) -- (v4);
\draw (v2) -- (v5);
\draw (v2) -- (v6);
\draw (v2) -- (v7);
\draw (v2) -- (v8);
\draw (v2) -- (v9);
\draw (v3) -- (v4);
\draw (v3) -- (v5);
\draw (v3) -- (v6);
\draw (v3) -- (v7);
\draw (v3) -- (v8);
\draw (v3) -- (v9);
\draw (v3) -- (v10);
\draw (v4) -- (v5);
\draw (v4) -- (v6);
\draw (v4) -- (v7);
\draw (v4) -- (v8);
\draw (v4) -- (v9);
\draw (v4) -- (v10);
\draw (v5) -- (v6);
\draw (v5) -- (v7);
\draw (v5) -- (v8);
\draw (v5) -- (v9);
\draw (v6) -- (v7);
\draw (v6) -- (v8);
\draw (v6) -- (v9);
\draw (v6) -- (v10);
\draw (v7) -- (v8);
\draw (v7) -- (v9);
\draw (v7) -- (v10);
\draw (v8) -- (v9);
\draw (v8) -- (v10);
\draw (v9) -- (v10);
    \end{tikzpicture}
    \caption{Graph representation of $\mathcal{I}_5$.}\label{fig:interval_graph}
\end{subfigure}
\caption{Illustration of the interval graph $\mathcal{I}_5$ in two different representations, where the only two independent sets of size larger than $1$ are depicted, one in red and one in cyan.}\label{fig:I_5}
\end{figure}
An example is illustrated in~\Cref{fig:I_5}. Naturally, the independent sets of $\mathcal{I}_n$ correspond to sets of intervals that do not intersect, and thus for which tasks with competing resources can be all scheduled.
Note that $\mathcal{I}_n$ has $\Omega(n^4)$ edges, as for each subset $\{i, j, k, \ell\} \subseteq \{1, \ldots, n\}$ with $i < j < k < \ell$ we have an edge between vertices $[i, k]$ and $[j, \ell]$. Therefore, a direct encoding of the independent-set property is very large even for small values of $n$. We will prove that this can be drastically improved to $\O(n^2 \lg n)$, but first we analyze how well clique coverings do for this family of graphs.

\begin{proposition}\label{prop:cliques_for_intervals}
    There is a CC-ISP encoding for $\mathcal{I}_n$ using $\O(n^3)$ clauses, and no CC-ISP encoding can be asymptotically more compact.
\end{proposition}
\begin{proof}
    First, we note that for every $k \in \{2, \ldots, n-1\}$, all the intervals $[i, j] \in V(\mathcal{I}_n)$ with $i \leq k \leq j$ intersect, thus forming a clique that we denote $K_{{\cap}k}$. Then, observe that the collection of cliques $K_{{\cap}k}$, for $2 \leq k \leq n-1$, is a clique covering of $\mathcal{I}_n$. Indeed, any edge $e = ([i, j], [a, b])$ must either have $a \leq j \leq b$, in which case $e$ is covered by $K_{{\cap} j}$, or $i \leq a \leq j$, in which case $e$ is covered by $K_{{\cap} a}$.
    Each clique $K_{{\cap} k}$ has $\O(n^2)$ vertices (there are $\O(n^2)$ vertices in the entire $\mathcal{I}_n$), and we thus this clique covering results in  $\sum_{k=2}^{n-1} |\amope(V(K_{{\cap} k} ))| = \O(n^3)$ many clauses. For the lower bound, consider an arbitrary clique covering $\mathcal{C}$ of $\mathcal{I}_n$. 
    Then, observe that each interval $x := [i, j]$ is adjacent to all the intervals $[i + 2t, i + 2t+1]$ for $t \in \{0, \ldots, \lfloor (j - i - 1)/2\rfloor\}$. Moreover, all these intervals $[i + 2t, i + 2t+1]$ are pairwise disjoint, which implies that for each $t$, the edge $\{x, [i + 2t, i + 2t+1]\}$ must be covered by a different clique $C_t \in \mathcal{C}$. 
    Thus, we have that $|\{C \in \mathcal{C} : x \in C\}| \geq \lfloor (j - i - 1)/2\rfloor + 1 \geq (j - i)/3$. We can now conclude since 
    \begin{align*}
    \sum_{C \in \mathcal{C}} |\amope(C)|  \geq \sum_{C \in \mathcal{C}} |C| = \sum_{x \in V(\mathcal{I}_n)} |\{C \in \mathcal{C} : x \in C\}| \geq \sum_{i = 1}^{n} \sum_{j = i+1}^{n} (j-i)/3,
    \end{align*}
    from where the change of variables $d := j-i$ yields
    \[
    \sum_{i = 1}^{n} \sum_{j = i+1}^{n} (j-i)/3 = \sum_{i = 1}^{n} \sum_{d = 1}^{n-i} d/3 
    \geq  \frac{1}{3}\sum_{i = \lfloor n/2\rfloor}^{n} \sum_{d = 1}^{\lfloor n/2 \rfloor} d  = \frac{1}{3}\lceil n/2 \rceil \cdot \frac{\lfloor n/2 \rfloor (\lfloor n/2\rfloor + 1)}{2} = \Omega(n^3).\qedhere
    \]
\end{proof}
\subsection{The Interval Propagation Trick}\label{subsec:ipt}
The proof of~\Cref{thm:intervals}, the main result of this section, is quite technical, and it is worth isolating one of its ingredients which might be of independent interest.
Consider the following encoding problem: 
\begin{center}
\emph{We have variables $x_{i, j}$, representing that an interval $[i, j]$ is ``selected'', for $1 \leq i < j \leq n$, and also variables $t_\ell$, for $1 \leq \ell \leq n$, whose intended semantics are that $t_\ell$ is true if and only if the index $\ell$ is contained in some selected interval. The problem is how to efficiently encode this relationship between the $x_{i, j}$ and $t_\ell$ variables, without enforcing any other conditions on either the $x$- or $t$-variables.} 
\end{center}
For example, if $x_{2, 4}$ and $x_{7, 9}$ are the only $x$-variables assigned to $\top$, then $\{t_2, t_3, t_4, t_7, t_8, t_9\}$ should be assigned to $\top$, and every other $t_\ell$ variable to $\bot$. 
The fact that $t_\ell$ implies that some interval containing $\ell$ is selected is trivial to encode, by just adding the $\O(n)$ following clauses:
\[
\overline{t_\ell} \lor \bigvee_{[i, j] \supseteq \{\ell\}} x_{i, j}, \quad \forall 1 \leq \ell \leq n.
\]
The other direction admits a nice trick.
The naïve way of encoding the implication from the $x$-variables toward the $t$-variables is to simply add clauses of the form 
\(
(\overline{x_{i, j}} \lor t_\ell), 
\)
for every $1 \leq i < j \leq n$ and every $i \leq \ell \leq j$, which amounts to 
\(
 \sum_{i = 1}^{n} \sum_{j = i+1}^{n} (j-i) = \Omega(n^3)
\)
many clauses, by the same analysis of the sum used in the proof of~\Cref{prop:cliques_for_intervals}.
It turns out, however, that we can achieve this with $\O(n^2)$ many clauses, using what we denote the \emph{``interval propagation trick''}.
First, we create variables $z_{i, j}$ for each $1 \leq i < j \leq n$, and then add the following clauses:
\begin{enumerate}
    \item $\overline{x_{i, j}} \lor z_{i, j}$, for every $1 \leq i < j \leq n$.
    \item $(\overline{z_{i, i+1}} \lor t_i)$ and $(\overline{z_{i, i+1}} \lor t_{i+1})$, for every $1 \leq i < n$.
    \item $(\overline{z_{i, j}} \lor z_{i+1, j})$ and $(\overline{z_{i, j}} \lor z_{i, j-1})$, for every $1 \leq i < j-1 < n$.
    \item $(\overline{z_{i, j}} \lor x_{i, j} \lor z_{i-1, j} \lor z_{i, j+1})$, for every $1 \leq i < j \leq n$, and removing the non-sensical literal $z_{i-1, j}$ when $i=1$, and $z_{i, j+1}$ when $j = n$.
\end{enumerate}

To formalize correctness, let us denote by $\textsf{NIP}_n$ the formula resulting from the aforementioned clauses in the naïve encoding (i.e., of the forms $\overline{t_\ell} \lor \bigvee_{[i, j] \supseteq \{\ell\}} x_{i, j}$ and $(\overline{x_{i, j}} \lor t_\ell)$), and 
$\textsf{IPT}_n$ the formula resulting from the clauses of the form $\overline{t_\ell} \lor \bigvee_{[i, j] \supseteq \{\ell\}} x_{i, j}$ together with the clauses of types (1-4) above. Note that $|\textsf{IPT}_n| \leq 6n^2 = \O(n^2)$, and let us now state the desired form of ``equivalence'' between these formulations.
\begin{restatable}{proposition}{iptcorrectness}
\label{prop:ipt_correctness}
    Let $\tau : \textsf{var}(\textsf{NIP}_n) \to \{\bot, \top\}$ and assignment. Then, we have that 
    \[
    \tau \models \textsf{NIP}_n \iff \textsf{SAT}(\textsf{IPT}_n|_\tau),
    \]
    and moreover, any satisfying assignment $\theta$ for $\textsf{IPT}_n|_\tau$ must assign \( \theta(z_{a, b}) = \top \) if and only if there is some $[i, j]$ such that $\tau(x_{i, j}) = \top$ and $[a, b] \subseteq [i, j]$ 
\end{restatable}
The proof of~\Cref{prop:ipt_correctness} is a rather tedious induction, and thus we defer it to~\Cref{sec:appendix}.

\subsection{\texorpdfstring{Better Encodings for $\mathcal{I}_n$}{Better Encodings for In}}
Our next result, which uses only $\O(n^2 \lg n)$ many clauses, requires a more careful encoding, which starts by decomposing $[1, n]$ into blocks of size at most $b$, which we do by assigning to each position $1 \leq i \leq n$ a block number $B(i) := \lceil i/b\rceil$. For now we will keep $b$ as a parameter, and denote by $k := \lceil n/b\rceil$ the number of blocks. 
We can characterize the different edges of $\mathcal{I}_n$ in terms of the blocks of their vertices, as the following lemma states.
\begin{restatable}{lemma}{edgecases}\label{lemma:edge_cases}
    Each edge $e := \{[i_1, j_1], [i_2, j_2]\} \in E(\mathcal{I}_n)$ with $i_1 \leq i_2$ must be part of exactly one of the following cases:
    \begin{enumerate}
        \item $B(i_1) = B(i_2)$ and $B(j_1) = B(j_2)$, and $i_2 \leq j_1$, in which case we say $e$ is an $x$-\emph{edge}.
        \item $B(i_1) < B(i_2) < B(j_1)$, in which case we say $e$ is a $y$-\emph{edge}.
        \item $B(i_1) = B(i_2)$ and $i_2 \leq j_1$, but $B(j_1) \neq B(j_2)$, in which case we say $e$ is an $s$-\emph{edge}.
        \item $B(i_1) < B(i_2) = B(j_1) = B(j_2)$ and $i_2 \leq j_1$, in which case we say $e$ is an $f$-\emph{edge}.
        \item $B(i_1) < B(i_2) = B(j_1) \neq B(j_2)$, and $i_2 \leq j_1$, in which case we say $e$ is an $m$-\emph{edge}.
    \end{enumerate}
    Moreover, any tuple $(i_1, j_1, i_2, j_2)$ with $i_1 \leq i_2$ that satisfies one of these cases implies $\{[i_1, j_1], [i_2, j_2]\} \in E(\mathcal{I}_n)$.
\end{restatable}

An illustration of~\Cref{lemma:edge_cases} is provided in~\Cref{fig:edge_cases}, and the proof is just case analysis and thus deferred to~\Cref{app:code}.


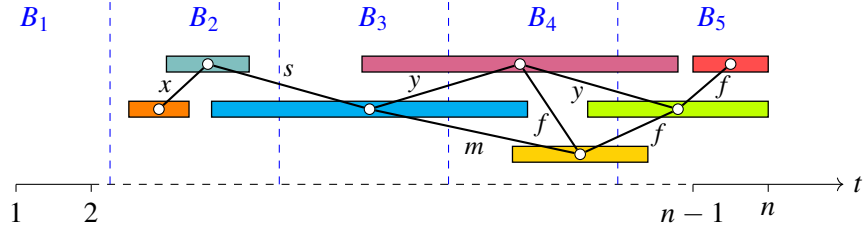
\begin{figure}
    \centering
    \begin{tikzpicture}
           \def\n{5}
           \def\secondheight{0.4}

           \def\downheight{0.6}
         \draw[-] (2,0-\downheight) -- (3,0-\downheight); 
         \draw[-, dashed] (3, 0-\downheight) -- (\n + 6, 0-\downheight); 
         \draw[->] (\n+6, 0-\downheight) -- (\n + 8, 0-\downheight) node[right] {$t$};
    \foreach \k in {1,...,2}
            \draw (\k+1,0-\downheight) -- (\k+1,-0.1-\downheight) node[below] {\(\k\)};
        
    \draw (12,0-\downheight) -- (12,-0.1-\downheight) node[below] {\(n\)};
    \draw (11,0-\downheight) -- (11,-0.1-\downheight) node[below] {\(n-1\)};

     \foreach \k in {1,...,4}
            \draw[blue, dashed] (1+2.25*\k,0-\downheight) -- (1+2.25*\k,2.5-\downheight); 
    
     \foreach \k in {1,...,5}
        \node[blue] at ( 2.25*\k, 2.2-\downheight) {\(B_\k\)};

  \setlength{\thicknesss}{6pt}    

  \pgfmathsetlengthmacro{\halfthick}{\thicknesss/2}

    \filldraw[ fill=teal!50!white, draw=black, line width=0.5pt]
        ($(4,1)-(0,\halfthick)$) rectangle ($(5.1,1)+(0,\halfthick)$);

      \filldraw[ fill=orange, draw=black, line width=0.5pt]
        ($(3.5,\secondheight)-(0,\halfthick)$) rectangle ($(4.3,\secondheight)+(0,\halfthick)$);

          \filldraw[ cyan, draw=black, line width=0.5pt]
        ($(4.6,\secondheight)-(0,\halfthick)$) rectangle ($(8.8,\secondheight)+(0,\halfthick)$);

                \filldraw[ purple!60!white, draw=black, line width=0.5pt]
        ($(6.6,1.0)-(0,\halfthick)$) rectangle ($(10.8,1.0)+(0,\halfthick)$);

                       \filldraw[red!70!white, draw=black, line width=0.5pt]
        ($(11, 1.0)-(0,\halfthick)$) rectangle ($(12,1.0)+(0,\halfthick)$);

                       \filldraw[lime, draw=black, line width=0.5pt]
        ($(9.6,\secondheight)-(0,\halfthick)$) rectangle ($(12,\secondheight)+(0,\halfthick)$);

                       \filldraw[ACMYellow, draw=black, line width=0.5pt]
        ($(8.6,\secondheight-0.6)-(0,\halfthick)$) rectangle ($(10.4,\secondheight-0.6)+(0,\halfthick)$);

        \node[draw, fill=white, circle, inner sep=0pt, minimum size=4pt] (teal_i) at (4.55, 1) {};
        \node[draw, fill=white, circle, inner sep=0pt, minimum size=4pt] (orange_i) at (3.9, \secondheight) {};
        \node[draw, fill=white, circle, inner sep=0pt, minimum size=4pt] (cyan_i) at (6.7, \secondheight) {};
        \node[draw, fill=white, circle, inner sep=0pt, minimum size=4pt] (purple_i) at (8.7, 1) {};
        \node[draw, fill=white, circle, inner sep=0pt, minimum size=4pt] (red_i) at (11.5, 1) {};
        \node[draw, fill=white, circle, inner sep=0pt, minimum size=4pt] (lime_i) at (10.8, \secondheight) {};
        \node[draw, fill=white, circle, inner sep=0pt, minimum size=4pt] (yellow_i) at (9.5, \secondheight-0.6) {};


        \draw[thick] (teal_i) -- (orange_i) node[midway, left] {\small \(x\)};
        \draw[thick] (teal_i) -- (cyan_i) node[midway, above] {\small \(s\)};
        \draw[thick] (cyan_i) -- (purple_i) node[midway, right, xshift=-18pt, yshift=1pt] {\small \(y\)};
        \draw[thick] (purple_i) -- (lime_i) node[midway,  below, xshift=-8pt, yshift=4pt] {\small \(y\)};
        \draw[thick] (lime_i) -- (red_i) node[midway, right] {\small \(f\)};

        \draw[thick] (cyan_i) -- (yellow_i) node[midway, below] {\small \(m\)};
        \draw[thick] (purple_i) -- (yellow_i) node[midway, left, yshift=-6pt, xshift=4pt] {\small \(f\)};
        \draw[thick] (lime_i) -- (yellow_i) node[midway, right, xshift=4pt, yshift=-1pt] {\small \(f\)};
    \end{tikzpicture}
    \caption{Illustration of~\Cref{lemma:edge_cases}. The type of each edge is indicated by its label, and blocks are separated by dashed blue lines.}\label{fig:edge_cases}
\end{figure}


\begin{theorem}\label{thm:intervals_formal}
    The independent-set property of $\mathcal{I}_n$ can be encoded using at most $26 n^2 \lg n$ clauses.
\end{theorem}
\begin{proof}
    \newcommand{\rnk}{\textsf{rnk}}
    We will prove a slightly stronger statement in order to have a stronger inductive hypothesis. Let $\mathcal{I}^0_n$ be the graph whose vertices are all the intervals $[i, j]$ for integers $1 \leq i < j \leq n$, but now with edges only between intervals whose intersection has cardinality at least $2$. That is, $\{[1, 3], [3, 5]\}$ is not an edge in $\mathcal{I}^0_5$, but it is in $\mathcal{I}_5$. 
    The proof is by (strong) induction on $n \geq 2$. The base case $n = 2$ is trivial since we can use the direct encoding then. In fact, it is easy to see that the direct encoding uses at most $3\binom{n}{4} = \frac{n^4}{8}$ clauses, and one can computationally check that for $n \leq 32$, 
    \(
        \frac{n^4}{8} \leq 26 n^2 \lg n,
    \)
    and thus the direct encoding is enough for the result. We thus assume $n > 32$ from now on.
    We now focus on the inductive case, where we will generally assume that we are encoding the independent-set property for $\mathcal{I}_n$, but indicate whenever a slight change is needed for $\mathcal{I}^0_n$, since the two cases are almost identical.
    The encoding will consider each of the four types of edges from~\Cref{lemma:edge_cases} separately.
    The base variables are $x_{i, j}$ (for $1 \leq i < j \leq n$), representing that the interval $[i, j]$ is part of the independent set.

    \begin{itemize}
        \item (\textbf{$x$-edges})  Consider a fixed choice of $B(i_1) = B(i_2) = \ell$ and $B(j_1) = B(j_2) = r$, and note that there will be $k^2$ such choices, since there are $k$ blocks in total. If $\ell = r$, then it is easy to see that the $x$-edges whose endpoints are in block $\ell$ form a graph isomorphic to $\mathcal{I}_{b}$ (resp. $\mathcal{I}^0_{b}$), and thus by inductive hypothesis they can be encoded using at most $26 b^2 \lg b$ clauses. If $\ell < r$, then we consider the graph $G_{\ell, r}$ formed by the $x$-edges whose endpoints are in blocks $\ell$ and $r$, even if they are both in $\ell$ or both in $r$. This time $G_{\ell, r}$ is isomorphic to $\mathcal{I}_{2b}$ (resp. $\mathcal{I}^0_{2b}$), and thus by the inductive hypothesis all these $x$-edges can be encoded using at most $26 (2b)^2 \lg(2b) = 104 b^2 \lg b + 104$ clauses. 
        As there are $k^2$ choices for $\ell, r$, we can encode all the $x$-edges of $\mathcal{I}_n$ using at most
        \(
        k^2 \cdot (104 b^2 \lg b + 104)
        \) clauses.
    

        \item (\textbf{$y$-edges}) We create, for each pair of block-indices $1 \leq \ell < r \leq k$, an auxiliary variable $y_{\ell, r}$ that represents that there is some interval $[i, j]$ in the independent set such that $B(i) = \ell$ and $B(j) = r$. To enforce these semantics, we add clauses 
    \begin{align}
        \overline{x_{i, j}} \lor y_{B(i), B(j)},& \quad \forall 1 \leq i < j \leq n, \label{eq:intro-y}\\
        \left(\overline{y_{\ell, r}} \lor \bigvee_{i, j : B(i) = \ell, B(j) = r} x_{i, j}\right),& \quad \forall 1 \leq \ell < r \leq k.\label{eq:intro-y2}
    \end{align}
       The key observation now is that the graph $G_y$ whose vertices are the $y_{\ell, r}$ variables and has edges between variables $y_{\ell_1, r_1}, y_{\ell_2, r_2}$ whenever $\ell_1 < \ell_2 < r_1$, is isomorphic to $\mathcal{I}^0_{k}$. 
       Therefore, by the inductive hypothesis, we can encode all the $y$-edges using $\binom{n}{2} + \binom{k}{2} + 26 k^2 \lg k$ many clauses.
      \item (\textbf{$s$-edges}) We create, for each block-index $2 \leq r \leq k$, and position $i$ such that $B(i) < r$, an auxiliary variable $s_{i, r}$ that represents that there is some interval $[i, j]$ in the independent set such that $B(j) = r$.
      We encode these semantics using clauses $(\overline{x_{i, j}} \lor s_{i, B(j)})$ and $(\overline{s_{i, r}} \lor \bigvee_{j \text{ s.t. } B(j) = r} x_{i, j})$, which amount to at most $2n^2$ clauses. 
       Then, we add clauses 
       \begin{equation}\label{eq:intro-s}
       \overline{s_{i_1, r_1}} \lor \overline{s_{i_2, r_2}}, \quad \forall i_1 \leq i_2, \text{ s.t. } B(i_1) = B(i_2), \forall r_1 > B(i_1), r_2 > B(i_1) \text{ with } r_1 \neq r_2.
       \end{equation}
        There are at most $b^2 \cdot k^3$ clauses from~\Cref{eq:intro-s}, since we need to choose $B(i_1), r_1, r_2$, for which there are $k^3$ possibilities, and then $b^2$ possibilities for $i_1, i_2$ such that $B(i_1) = B(i_2)$.
        Unfortunately, this is not enough to encode all $s$-edges, since~\Cref{eq:intro-s} misses the cases where one of the two intervals is entirely contained in one block, so either $B(i_1) = B(j_1)$ or $B(i_2) = B(j_2)$.
        The naïve solution would be to add the following clauses:
        \begin{align}
        \overline{x_{i_1, j_1}} \lor \overline{s_{i_2, r}},& \quad \forall 1 \leq i_1 \leq i_2 \leq n, \forall j_1 > i_1 \text{ s.t. } B(i_1) = B(j_1) = B(i_2), \forall r > B(i_2),\label{eq:intro-s2}\\
        \overline{x_{i_2, j_2}} \lor \overline{s_{i_1, r}},& \quad \forall 1 \leq i_1 \leq i_2 \leq n, \forall j_2 > i_2 \text{ s.t. } B(i_1) = B(i_2) = B(j_2), \forall r > B(i_1).\label{eq:intro-s3}
        \end{align}
        \Cref{eq:intro-s2} uses at most $b^3 \cdot k^2$ clauses, since we need to choose $B(i_1) = B(j_1) = B(i_2)$ and $r$, and for each such choice there are at most $b^3$ possibilities for $i_1, i_2, j_1$ within the same block.
        Analogously,~\Cref{eq:intro-s3} also incurs in at most $b^3 \cdot k^2$ clauses.
        However, it will turn out that $b^3 \cdot k^2$ clauses would be too large, since we will end up setting $k = \Theta(\lg n)$, and thus we need a slightly better way to encode~\Cref{eq:intro-s2,eq:intro-s3}.
        The solution is to use the ``interval propagation trick'' from~\Cref{subsec:ipt} independently in each block of index $1 \leq d \leq k$, thanks to which we can assume variables $t^{d}_\ell$ that represent whether some $x_{i, j}$ with $\ell in [i, j]$ is true with $B(\ell) = B(i) = B(j) = d$ using at most $6b^2$ clauses. In total over the $k$ blocks this incurs in at most $6b^2 \cdot k$ clauses.
        Now, wen can replace~\Cref{eq:intro-s2,eq:intro-s3} by 
        \begin{equation}\label{eq:intro-s2-t}
        \overline{t^{B(i)}_{\ell}} \lor \overline{s_{i, r}}, \quad \forall 1  \leq i \leq \ell \leq n,  \text{ s.t. } B(\ell) = B(i), \forall r > B(i).
        \end{equation}
        \Cref{eq:intro-s2-t} only requires $k \cdot b^2$ many clauses, since there are at most $k$ choices for $r$, and $b^2$ for $i, \ell$.
        We thus cover all $s$-edges using a total of $2n^2 + b^2  k^3 + 7 b^2  k$ clauses.
        

       \item (\textbf{$f$-edges})  This case is fully symmetrical to the $s$-edges, this time using variables $f_{\ell, j}$ that represent the presence of some interval $[i, j]$ in the independent set such that $B(i) = \ell$.
       We add the symmetrical clauses, e.g., the symmetrical of~\Cref{eq:intro-s} is
    \[
       \overline{f_{\ell_1, j_1}} \lor \overline{f_{\ell_2, j_2}}, \quad \forall 1 \leq j_1, j_2 \leq n, \text{ s.t. } B(j_1) = B(j_2), \forall \ell_1 < B(j_1), \ell_2 < B(j_1) \text{ with } \ell_1 \neq \ell_2.    
       \]
       and similarly with the symmetrical of~\Cref{eq:intro-s2-t}. A minor saving is that we do not need to pay any extra clauses for having the variables $t^d_{i}$, 
       and therefore the combination of this case with the $s$-edges incurs in a total of $4n^2 + 2b^2 k^3 + 8b^2 k$ clauses. 
        \item (\textbf{$m$-edges}) In this case we are dealing with intervals $[i_1, j_1]$ and $[i_2, j_2]$ such that $B(j_1) = B(i_2)$, which we call $d := B(j_1)$, and the other two endpoints not being in the block $d$. We can cover these by using both our $s$ and $f$ variables.
       Indeed, it suffices to add clauses 
       \begin{equation}
        \overline{f_{\ell, j_1}} \lor \overline{s_{i_2, r}}, \quad \forall 1 \leq i_2 \leq j_1 \leq n \text{ s.t. } B(j_1) = B(i_2), \forall \ell < B(j_1), \forall r > B(i_2),
       \end{equation}
       which amounts to at most $b^2 k^3$ clauses, since we have to choose $\ell, r, B(j_1)$, for which there are $k^3$ options, and conditioned on $B(j_1)$ there are at most $b^2$ choices for $j_1, i_2$.
    \end{itemize}

    Adding the total number of clauses over all types, we get a total of
    \begin{align*}
        k(104b^2 \lg b + 104) + \binom{n}{2} + \binom{k}{2} + 26k^2 \lg k + 4n^2 + 2b^2k^3 + 8b^2k + b^2k^3 
    \end{align*}
    clauses. Using $n = kb$, this is at most
    \(
        104nb \lg b  + 104k + 4.5n^2 + 0.5k^2 + 26k^2 \lg k + 3n^2k + 8nb,
    \)
    and then taking $k = \lfloor\lg n\rfloor$ and $b = n/\lfloor\lg n\rfloor$, this is at most
    \[
        3n^2\lg n + n^2\left(108.5 + \frac{\lg n + 8.5\lg^2(n) \lg\lg n}{n^2} + \frac{8}{\lg n}\right) \leq 3n^2\lg n + 109.5n^2 \tag{since $n \geq 32$},  
    \]
    but as $n > 32$, we have $\lg n > 5$, and thus $109.5n^2 < 22n^2 \lg n$, from where $3n^2 \lg n + 109.5n^2 < 25n^2\lg n$, and thus we conclude our result.
\end{proof}

\subsection{Applications to Scheduling Problems}\label{subsec:applications}

We consider the \emph{non-preemptive schedulability} question treated by~\citet{mayankEfficientSATEncoding2020}, where there are $N$ tasks, the $i$-th of which has an integer duration $d_i$ and must be both started and finished in the interval $[r_i, e_i]$, with $1 \leq r_i \leq e_i \leq T$, and moreover, there are $M$ machines which can do tasks in parallel. They present several SAT encodings, all of which use $\Omega(N\!MT^2)$ clauses~\citep[Table 2]{mayankEfficientSATEncoding2020}. \Cref{thm:intervals_formal} allows us to do better:
\begin{theorem}[Informal]\label{thm:non_preemptive_scheduling}
    The non-preemptive schedulability problem can be encoded using $\O(N\!MT + M T^2 \lg T)$ clauses.
\end{theorem}
\begin{proof}
Create variables $x_{i, t, m}$ that represents that task $i$ is assigned to start in machine $m$ at time $t$, and auxiliary variables $y_{m, t_1, t_2}$ that represent that there is some task assigned to machine $m$ for exactly the interval $[t_1, t_2]$.
The semantics of the $y$-variables are encoded by clauses 
\[
    \overline{x_{i, t, m}} \lor y_{m, t, t + d_i}, \quad \forall i \in [1, N], t \in [r_i, e_i], m \in [1, M],
\]
using $\O(N\!MT)$ clauses. We partition the tasks according to their duration, with $D_t = \{ i : d_i = t\}$, so that we can enforce 
$\amope(\{x_{i, t', m} : i \in D_t\})$, for each $t' \in [1, T], t \in [1, T - t']$ and $m \in [1, M]$, which uses 
\[
    \sum_{t' = 1}^T \sum_{t \in [1, T - t']} \O(|D_t|) \cdot M \leq M T \cdot \O\left(\sum_{t \in [1, T - t']} |D_t| \right) = \O(N\!M T)
\]
many clauses. 
Then, for each index $i$, we enforce that task $i$ is done at some point, in some machine, with $\O(N)$ clauses: 
\[
\bigvee_{\substack{m \in [1, M]\\
    t \in [r_i, e_i-d_i]}} x_{i, t, m}, \quad \forall i \in [1, N],
\]
We then use, independently for each $m$, the encoding of~\Cref{thm:intervals_formal} to encode that the variables $y_{m, t_1, t_2}$ assigned to true make for disjoint intervals. 
This results in $\O(N\!MT + M T^2 \lg T)$ clauses. Correctness follows from the facts (i) we explicitly enforce that each task is done at some point, in some machine, and respecting its time constraints, (ii) the $\amope$ constraints ensure that no two tasks are assigned to the same machine during the same time interval, and (iii) the \emph{disjoint intervals} encoding on the $y_{m, t_1, t_2}$ variables ensures that no machine is used for two overlapping time intervals.
 that by the $\amope$ constraints we are forbidding different  
\end{proof}
When $N$ and $M$ are $\Theta(T)$, \Cref{thm:non_preemptive_scheduling} results in $\O(T^3 \lg T)$ clauses, as opposed to the $\Omega(T^4)$ clauses of~\citet{mayankEfficientSATEncoding2020}.
It is worth noting that, as done for example by~\citet{maric}, an alternative option would be to add for each time $t$, and machine $m$, a constraint 
\[
    \amope(\{x_{i, t', m} : i \in [N], t' \in [t-d_i, t]\}),
\]
however $\{x_{i, t', m} : i \in [N], t' \in [t-d_i, t]\}$ is a set of size $\Omega(NT)$, and thus this would result in $\Omega(N\!MT^2)$ clauses. 
%

\section{Discussion and Future Work}\label{sec:conclusions}

We have proposed a theoretical framework that will hopefully be helpful in the quest for understanding the limits of CNF encodings. To do so, we considered ``covering encodings'', based on covering the graph of incompatible pairs of literals with either cliques or bicliques, leading to different encodings. We showed how both clique coverings and biclique coverings have different advantages, where clique coverings are more efficient in the complete graph but biclique coverings are more efficient in the worst-case (\Cref{thm:subquadratic_encoding}). This difference is essentially surmounted for random Erd\H{o}s-Renyi graphs (\Cref{prop:random_coverings}). 
Moreover, it is worth noting that clique coverings are also very efficient when the graph is ``very close to being complete'', meaning that every vertex has degree at least $n - \Theta(1)$, as in this case a nice result of~\citet{alonCoveringGraphsMinimum1986} gives a covering with $\O(\lg n)$ cliques and thus an encoding with $\O(n \lg n)$ clauses.
We have shown a modest lower bound in~\Cref{thm:lower_bound}, which only applies to encodings using constant-width clauses. Extending this to general encodings is an interesting direction for future work.
Even though our study here has been theoretical in nature, I have implemented clique-covering and biclique-covering algorithms. 
In particular, I tested the algorithm of~\citet{mubayiFindingBipartiteSubgraphs2010} for biclique coverings with the guarantee of~\Cref{thm:erdos_chung_spencer}, however, the algorithm was designed to prove its asymptotic quality and gives poor results for small $n$. Therefore, a natural next step is to design a more practically efficient algorithm to find good biclique coverings.
For clique coverings, I tested both using the greedy approach of selecting the maximum clique at a time (for which I used the \texttt{Cliquer} tool~\citep{aaltoCliquerHomepage}), and the specific clique-covering tool from~\citet{algCliqueCovering}. The resulting quality of the coverings was not too different, but the latter algorithm was orders of magnitude faster for formulas with $\approx \! 1000$ variables.
A more thorough experimental evaluation is in order, especially considering the several more modern maximum clique algorithms, with their different trade-offs~\citep{wuReviewAlgorithmsMaximum2015}.
In terms of related work to our covering ideas, besides its nearest neighbors being the works of~\citet{Rintanen} and~\citet{ignatievCardinalityEncodingsGraph2017}, we highlight that~\citet{juknaComputationalComplexityGraphs2013} provides an in-depth treatment of the relationship between biclique coverings and the complexity of formulas representing graphs, especially bipartite ones.
Rintanen seems to be the earliest occurrence of the clique/biclique covering idea, although it is worth noting that his work explicitly states that his biclique representation still yields $\Omega(n^2)$ clauses for graphs of $n$ vertices~\citep[Section 6]{Rintanen}.

We have also studied how our framework applies to the case of complete interval graphs which are useful for encoding planning and scheduling problems. 
In fact, the~\emph{primitive} of encoding disjoint intervals seems to be useful in diverse contexts: a prior version of our encoding, that uses $\O(n^{8/3}) = \O(n^{2.666\ldots})$ clauses, was successfully used to improve an $\O(n^4)$ encoding for \emph{Straight Line Programs} (SLPs) from~\citet{bannai_et_al:LIPIcs.ESA.2022.12}, where a slight variant of the disjoint-intervals property was the bottleneck, and the improved encoding led to a total of $\O(n^3)$ clauses, making a different constraint the bottleneck. The improved SLP encoding incorporating our ideas is currently under review~\citep{bannaiAndMe}. 
To obtain $\O(n^{8/3})$ clauses, the encoding is essentially the same as the one in~\Cref{thm:intervals_formal}, but except of proceding recursively for the $x$-edges and the $y$-edges, we encode those directly. 
This leads to $\O(k^2 b^4)$ for the $x$-edges, since for any pair of blocks $\ell, r$ (of which there are $k^2$), we forbid the $x$-edges between intervals $[i, j]$ and $[i', j']$ such that $B(i) = B(i') = \ell$ and $B(j) = B(j') = r$, and there are at most $b^4$ choices for $i, j, i', j'$ conditioned on their blocks being $\ell$ and $r$.
For the $y$-edges, this leads to $\O(k^4)$ clauses, since we need to forbid the $y$-edges between pairs of blocks that guarantee an interval overlap.
Therefore, the total number of clauses is 
\(
    \O(k^2 b^4 + k^4 + n^2 + k^3 b^2), 
\)
for which a simple calculus argument shows that the optimal choice is $k = n^{2/3}$ and $b = n^{1/3}$, leading to the number of clauses being 
\[
\O(n^{4/3} n^{4/3} + n^{8/3} + n^2 + n^{6/3} n^{2/3}) = \O(n^{8/3} + n^{8/3} + n^2 + n^{8/3}) = \O(n^{8/3}).
\]


The difference between the $\O(n^{8/3})$ encoding and the result in~\Cref{thm:intervals_formal} is the use of recursion. Indeed, the $\O(n^{8/3})$ encoding can be formulated as a carefully crafted biclique covering, making for another example of the difference between BVA-style encodings, which allow for covering auxiliary variables, and encodings that only cover the base variables.
Running BVA on top of the $\O(n^{8/3})$ encoding led to smaller encodings that either of them alone, reenforcing the idea that BVA can operate recursively on top of an initial covering (see~\Cref{tab:encoding-comparison2} in~\Cref{app:disjoint-intervals}). 
 Experimental results for the SLP problem of~\citet{bannai_et_al:LIPIcs.ESA.2022.12} are presented in~\Cref{tab:encoding-comparison} (\Cref{app:disjoint-intervals}).

 The recursion of blocks in~\Cref{thm:intervals_formal} has a nice interpretation for scheduling: if one were to schedule events on e.g., a calendar year, starting on day $d_1$ and ending on day $d_2$, it would be convenient to first catalogue the events based on their starting and ending month: \emph{anything starting in January and ending in May is incompatible with anything starting in March and ending in September.} Then, for more granularity, the same technique decomposes months into weeks, and so on. The concrete decomposition in~\Cref{thm:intervals_formal} used $\lg n$ blocks in the first recursive step, which would be $\lg 365 \approx 8.5$ blocks for a calendar year, on the order of magnitude of months. On the other hand, the decomposition for the $\O(n^{8/3})$ version, where only one level of recursion is used, takes $k = n^{2/3}$, which would be roughly $365^{2/3} \approx 51$ for a calendar year, so basically the number of weeks.
Studying the practical applicability of our encoding for scheduling problems is an interesting direction for future work. In general, it is not necessarily true that fewer clauses will lead to a reduced solving time~\citep{DBLP:journals/jsat/Bjork11,haberlandtEffectiveAuxiliaryVariables2023}. For example, together with Marijn Heule, we presented an $\O(n^2 k \lg k)$ encoding for the packing chromatic number of~$\mathbb{Z}^2$~\citep{subercaseauxPackingChromaticNumber2022}, which ended up being surpassed by a much more effective although asymptotically larger encoding~\cite{subercaseauxPackingChromaticNumber2023d}.

Finally, we note that our ideas can be readily applied to \emph{Integer Linear Programming} formulations, and hopefully to other forms of constraint programming too.

\bibliographystyle{plainnat}
\bibliography{references}

\appendix

\section{Proof of~\texorpdfstring{\Cref{prop:ipt_correctness}}{Proposition 19}}\label{sec:appendix}

\iptcorrectness*
\begin{proof}
    For the $(\implies)$ direction, we assume that $\tau \models \textsf{NIP}_n$, and the build the assignment $\theta : \textsf{var}(\textsf{IPT}_n|_\tau) \to \{\bot, \top\}$ as described in the statement of the proposition:
    \( \theta(z_{a, b}) = \top \) if and only if there is some $[i, j]$ such that $\tau(x_{i, j}) = \top$ and $[a, b] \subseteq [i, j]$.
    Then the clauses (1-4) can be easily checked to be satisfied by $\theta$.  
        For the $(\impliedby)$ direction, we assume a satisfying assignment $\theta$ for $\textsf{IPT}_n|_\tau$ and assume expecting a contradiction that $\tau \nvDash \textsf{NIP}_n$. Since $\textsf{IPT}_n|_\tau$ is satisfiable, and $\textsf{IPT}_n$ contained the clauses $\overline{t_\ell} \lor \bigvee_{[i, j] \supseteq \{\ell\}} x_{i, j}$ whose variables were assigned by $\tau$, the only possibility is that 
    \(
    \tau \nvDash (\overline{x_{i, j}} \lor t_\ell)
    \) for some $1 \leq i < j \leq n$ and $\ell \in [i, j]$.  
    Thus, we have that $\tau(x_{i, j}) = \top$ and $\tau(t_\ell) = \bot$.  
    But by the clauses of type (1), $\tau(x_{i, j}) = \top$ implies $\theta(z_{i, j}) = \top$, so it suffices to prove that $\theta(z_{i, j}) = \top$ contradicts $\tau(t_\ell) = \bot$. We do this by showing that $\theta(z_{i, j}) = \top$ implies $\tau(t_\ell) = \top$, for any $\ell \in [i, j]$. The proof is by induction over $d := j-i$. If $d = 1$,
    then the clauses of type (2) directly give us $\tau(t_\ell) = \top$, and if $d > 1$, the clauses of type (3) give us that $\theta(z_{i+1, j}) = \top$ and $\theta(z_{i, j-1}) = \top$, and as $\ell$ belongs to either $[i+1, j]$ or $[i, j-1]$, we conclude by the inductive hypothesis.

    Let us now show that no other $\theta$ works. Indeed, we first prove that $\tau(x_{i, j}) = \top$ implies that \( \theta(z_{a, b}) = \top \) for every $[a, b] \subseteq [i, j]$, by induction on $d := (a-i) + (j-b)$. If $d = 0$, then $a = i$ and $j = b$, so by the clause of type (1) we are done.
     Otherwise, we need to prove that the statement for $d$ implies the case for $d+1$ for any $d < j - i - 1$ (since $j-i-1$ is the maximum possible value of $d$).
     We assume $\tau(x_{i, j}) = \top$ and note that by the inductive hypothesis this implies $\theta(z_{a, b}) = \top$ for any $1 \leq a < b \leq n$ such that $(a-i) + (j-b) = d$. Then, any interval $[a', b'] \subseteq [i, j]$ with $(a'-i) + (j-b') = d+1$ must be of the form $[a+1, b]$ or $[a, b-1]$ with $(a-i) + (j-b) = d$, and as by the clauses of type (3) we have $\theta(z_{a+1, b}) = \top$ and $\theta(z_{a, b-1}) = \top$, we are done.
     On the other hand, assume that \( \theta(z_{a, b}) = \top \) for some pair $1 \leq a < b \leq n$, and let us show that $\tau(x_{i, j}) = \top$ for some $[i, j]\supseteq [a, b]$.
     This time the induction is over $d := n - (b-a)$, with the base case $d = 1$ implying that $a=1, b = n$, from where (4) directly yields $z_{1, n} \to x_{1, n}$ which proves the base case. For the inductive case,
     note that a clause of type (4) guarantees that either $\tau(x_{a, b}) = \top$, in which case we are done, or that
      either $\theta(z_{a-1, b}) = \top$ or $\theta(z_{a, b+1})$ hold. But $n - (b - (a-1)) = n - (b+1 - a) = d-1$, and thus by inductive hypothesis we have that $\tau(x_{i, j}) = \top$ for some $[i, j]$ such that $[i, j] \supseteq [a-1, b] \supset [a, b]$ or $[i, j] \supseteq [a, b+1] \supset [a, b]$.
\end{proof}

\section{Proof of~\texorpdfstring{\Cref{lemma:edge_cases}}{Lemma 20}}\label{app:code}

\edgecases*
\begin{proof}
First, observe that in all cases we are assuming $i_1 \leq i_2$, and thus as all cases except (2) explicitly require $i_2 \leq j_1$, they all imply the intersection of the intervals $[i_1, j_1]$ and $[i_2, j_2]$ is non-empty.
For case (2), note that $B(i_2) < B(j_1)$ implies $i_2 < j_1$, which together with $i_1 \leq i_2$ implies again that the intersection of the intervals $[i_1, j_1]$ and $[i_2, j_2]$ is non-empty.
Thus, all cases correspond to actual edges in the graph $\mathcal{I}_n$. 
To see exhaustiveness, we present in~\Cref{fig:decision-tree} a decision tree that cases on the relationships between the blocks.
\end{proof}
Furthermore, as a sanity check, we present in~Code~\ref{lst:check_edge_cases} a Python script that validates the correctness for some finite values.

\renewcommand{\lstlistingname}{Code}
\begin{lstlisting}[language=Python, caption={Code to validate~\Cref{lemma:edge_cases}.}, captionpos=t, label={lst:check_edge_cases}]
def B(i, b):
    # ceil(i/b)  can be written as (i + b - 1)//b
    return (i + b - 1) // b

def classify(i1, j1, i2, j2, b):
    """Return which case (1,2,3,4) the edge {[i1,j1],[i2,j2]} belongs to."""
    bi1, bi2 = B(i1,b), B(i2,b)
    bj1, bj2 = B(j1,b), B(j2,b)

    # the five predicates from the lemma
    is_x = (bi1 == bi2) and (bj1 == bj2)
    is_y = (bi1 < bi2)  and (bi2 < bj1)
    is_s = (bi1 == bi2) and (bj1 != bj2)
    is_f = (bi1 < bi2)  and (bj1 == bj2)   and (bi2 == bj1)
    is_m = (bi1 < bi2) and (bi2 == bj1) and (bj1 != bj2)  

    flags = [is_x, is_y, is_s, is_f , is_m]
    if sum(flags) < 1:
        # either none or more than one case matched
        return flags
    return flags.index(True) + 1  # return 1,2,3,4, or 5.

def check_lemma(n, b):
    bad = []
    for i1 in range(1, n+1):
        for i2 in range(i1, n+1):
            for j1 in range(i1+1, n+1):
                for j2 in range(i2+1, n+1):
                    # only look at intersecting intervals
                    if i2 <= j1:
                        cls = classify(i1, j1, i2, j2, b)
                        if cls is None or type(cls) is not int:
                            bad.append((i1,j1,i2,j2))
    if bad:
        print("Found unclassified or multiply-classified edges:")
        for quad in bad:
            print(quad, "-> (blocks)", [B(q, b) for q in quad])
    else:
        print(f"All edges classified correctly! (n={n},b={b})")

if __name__ == "__main__":
    # Example for n = 50, block size of 10.
    check_lemma(n=50, b=10)

    # Example without exact divisibility
    check_lemma(n=50, b=8)
\end{lstlisting}

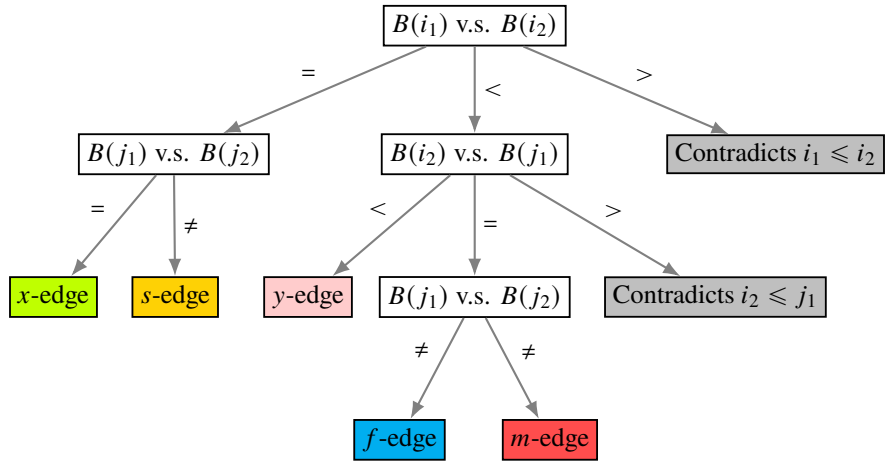
\begin{figure}
    \centering
\begin{tikzpicture}[
    xscale=0.8,
    yscale=0.95,
    level 1/.style={sibling distance=50mm, level distance=18mm},
    level 2/.style={sibling distance=40mm, level distance=20mm},
    level 3/.style={sibling distance=25mm, level distance=20mm},
    level 4/.style={sibling distance=15mm, level distance=20mm},
    level 5/.style={sibling distance=9mm, level distance=20mm},
    every node/.style={draw=black, line width=0.7,  inner sep=3pt, font=\small},
    edge from parent/.style={thick, draw=black!50!white, -{Latex}},
    edge label/.style={draw=none, fill=none, font=\footnotesize, auto},
    highlight path/.style={edge from parent/.append style={draw=blue, thick}}
]

\node[draw] {$B(i_1) \text{ v.s. } B(i_2)$} 
    child {
        node[draw] {$B(j_1) \text{ v.s. } B(j_2)$}
        child {
            node[draw, fill=lime] {$x$-edge}
            edge from parent node[ right, edge label, auto, xshift=-12pt, yshift=11pt] {$=$}
        }
        child {
            node[draw, fill=ACMYellow, xshift=-45pt] {$s$-edge}
            edge from parent node[ above, edge label] {$\neq$}
        }
        edge from parent node[below, edge label, xshift=-12pt, yshift=11pt] {${=}$}
    }
    child {
        node[draw] {$B(i_2) \text{ v.s. } B(j_1)$}
        child {
            node[draw, fill=red!20, xshift=28pt] {$y$-edge}
            edge from parent node[ above, edge label, xshift=-12pt, yshift=12pt] {$<$}
        }
        child {
            node[draw] {$B(j_1) \text{ v.s. } B(j_2)$}
            child {
                node[draw, fill=cyan] {$f$-edge}
                edge from parent node[ above, edge label, xshift=-12pt, yshift=13pt] {$\neq$}
            }
            child {
                node[draw, fill=red!70!white] {$m$-edge}
                edge from parent node[ above, edge label] {$\neq$}
            }
            edge from parent node[ above, edge label] {$=$}
        }
        child {
            node[draw, fill=lightgray] {Contradicts $i_2 \leq j_1$}
            edge from parent node[ below, edge label] {$>$}
        }
        edge from parent node[ above, edge label] {${<}$}
    }
    child {
        node[draw, fill=lightgray] {Contradicts $i_1 \leq i_2$}
        edge from parent node[below, edge label] {$>$}
    };

\end{tikzpicture}
\caption{A decision tree for the cases of~\Cref{lemma:edge_cases}. }\label{fig:decision-tree}
\end{figure}

\section{Application to String Compression with SLPs}\label{app:disjoint-intervals}

Note that in the problem of ~\citet{bannai_et_al:LIPIcs.ESA.2022.12}, the intervals can overlap if one is strictly contained in the other. This can be achieved by minor modifications in the indices of our constraints, without any new conceptual ideas.

We display some preliminary experimental results in~\Cref{tab:encoding-comparison}, where we show the usage of the $\O(n^{8/3})$ encoding for disjoint intervals in replacement of constraint 8 of the SLP encoding of ~\citet{bannai_et_al:LIPIcs.ESA.2022.12}. Our experiments are over families of string that make standard examples for string compression, and their descriptions can be found in the paper of ~\citet{bannai_et_al:LIPIcs.ESA.2022.12}, and the concrete strings are publicly available in their repository: \url{https://github.com/kg86/satcomp}.
Perhaps the most striking example is that of Fibonacci binary strings, defined recursively as $F_0 = 0$, $F_1 = 1$, and $F_n = F_{n-1} F_{n-2}$ (concatenation) for $n \geq 2$, where as shown in~\Cref{tab:encoding-comparison}, the improved encoding reduces the number of clauses from $118$ millions to $12$ millions (\texttt{fib12.txt}).
Generally speaking, the impact of the encoding for the disjoint (or strictly contained) intervals is not always as large as the asymptotics suggest, since the constraint in the encoding from \cite[Equation (8)]{bannai_et_al:LIPIcs.ESA.2022.12} are only applied based on a condition that depends on the repetitiveness of the input string. The impact of the encoding is thus maximal for strings of the form $a^n$, for some symbol $a$. These strings are interesting from a theoretical point of view, since the complexity of the smallest SLP (which is NP-hard to compute over arbitrary strings) is not known~\cite{171713}.

\begin{table}[htbp]
  \centering
  \caption{Comparison of encoding methods for the independent-set property of $\mathcal{I}_n$. We use the acronym `bd' for the $\O(n^{8/3})$ encoding, standing for ``block decomposition''.}
  \label{tab:encoding-comparison2}
  \begin{tabular}{crrrrrrrrr}
    \toprule
    & \multicolumn{2}{c}{naïve} & \multicolumn{2}{c}{bva} & \multicolumn{2}{c}{bd} & \multicolumn{2}{c}{bd+bva} \\
    \cmidrule(lr){2-3} \cmidrule(lr){4-5} \cmidrule(lr){6-7} \cmidrule(lr){8-9}
    $n$ & vars & clauses & vars & clauses & vars & clauses & vars & clauses \\
    \midrule
    10 & 55 & 330 & 82 & 143 & 130 & 196 & 132 & 194 \\
    20 & 210 & 5,985 & 402 & 862 & 398 & 984 & 430 & 801 \\
    30 & 465 & 31,465 & 1,015 & 2,499 & 850 & 2,714 & 952 & 1,893 \\
    40 & 820 & 101,270 & 1,935 & 4,996 & 1,315 & 5,774 & 1,687 & 3,095 \\
    \bottomrule
  \end{tabular}
\end{table}

\begin{table}
  \centering
  \caption{Comparison of encodings for smallest SLPs}
  \begin{tabular}{lrrrrrrr}
    \toprule
    \multirow{2}{*}{File} & \multirow{2}{*}{$|T|$} & \multicolumn{3}{c}{Base Encoding} & \multicolumn{3}{c}{Base + $\O(n^{8/3})$ disjoint intervals} \\
    \cmidrule(lr){3-5} \cmidrule(lr){6-8}
    & &  Time (s) & Clauses & Vars &  Time (s) & Clauses & Vars \\
    \midrule
    \multicolumn{8}{l}{} \\
    fib6.txt & 21 & \textbf{0.15} & 3,754 & 494 & 0.25 & 3,808 & 719 \\
    fib7.txt & 34 & \textbf{0.93} & 16,764 & 1,513 & 1.39 & 14,465 & 1,987 \\
    fib8.txt & 55 & \textbf{3.67} & 84,554 & 4,716 & 3.88 & 52,746 & 5,716 \\
    fib9.txt & 89 & \textbf{5.14} & 472,269 & 14,604 & 5.17 & 203,392 & 16,683 \\
    fib10.txt & 144 & 6.97 & 2,849,938 & 44,848 & \textbf{6.57} & 798,982 & 49,430  \\
    fib11.txt & 233 & 33.16 & 18,101,282 & 135,859 & \textbf{20.53} & 3,119,598 & 145,687  \\
    fib12.txt & 377 & 214.92 &  118,647,206 & 406,448 & \textbf{32.45} & 12,363,254 &  428,235\\

    \multicolumn{8}{l}{} \\
    thue\_morse5.txt & 32 & \textbf{0.40} & 9,562 & 948 & 0.70 & 11,199 & 1,398  \\
    thue\_morse6.txt & 64 & \textbf{3.48} & 73,015 & 4,339 & 4.83 & 71,349 & 5,563  \\
    thue\_morse7.txt & 128 & \textbf{8.15} & 656,099 & 19,933 & 8.18 & 532,394 & 23,740  \\
    thue\_morse8.txt & 256 & 12.60 & 6,977,803 & 90,647 & \textbf{11.61} & 3,825,669 & 102,260  \\
    \multicolumn{8}{l}{} \\
    period\_doubling5.txt & 32 & \textbf{0.49} & 11,682 & 1,135 & 0.79 & 11,962 & 1,585  \\
    period\_doubling6.txt & 64 & \textbf{4.36} & 109,928 & 5,586 & 4.80 & 75,780 & 6,810  \\
    period\_doubling7.txt & 128 & 7.60 & 1,281,958 & 27,219 & \textbf{7.31} & 557,642 & 31,026  \\
    period\_doubling8.txt & 256 & 28.23 & 17,372,984 & 130,050 & \textbf{12.41} & 3,951,407 & 141,663  \\
    \multicolumn{8}{l}{} \\
    paperfold4.txt & 32 & \textbf{0.23} & 9,078 & 880 & 0.55 & 10,869 & 1,330  \\
    paperfold5.txt & 64 & \textbf{1.66} & 67,650 & 3,924 & 3.71 & 69,736 & 5,148  \\
    paperfold6.txt & 128 & \textbf{6.14} & 598,774 & 17,613 & 8.22 & 523,842 & 21,420  \\
    paperfold7.txt & 256 & 14.17 & 6,280,334 & 78,705 & \textbf{12.92} & 3,782,468 & 90,318  \\
    \bottomrule
  \end{tabular}
  \label{tab:encoding-comparison}
\end{table}

\end{document}

%% file: biclique_decomp.tex
\centering
\begin{tikzpicture}[
    bipartitenode/.style={circle,draw,minimum size=13pt, inner sep=0pt},
    clausenode/.style={},
    edge/.style={-, line width=2.5pt},
]
\def\uphigh{0.3}
  \node[bipartitenode, fill=white, text=black] (u_left_0) at (0.0cm,-0.0cm) {1};
  \node[bipartitenode, fill=white, text=black] (u_left_1) at (0.0cm,-0.9cm) {2};
  \node[bipartitenode, fill=white, text=black] (u_left_2) at (0.0cm,-1.8cm) {3};
  \node[bipartitenode, fill=white, text=black] (u_left_3) at (0.0cm,-2.7cm) {4};
  \node[bipartitenode, fill=white, text=black] (u_left_4) at (0.0cm,-3.6cm) {5};
  \node[bipartitenode, fill=black, text=white] (v_right_5) at (4.0cm,-0.0cm) {6};
  \node[bipartitenode, fill=black, text=white] (v_right_6) at (4.0cm,-0.9cm) {7};
  \node[bipartitenode, fill=black, text=white] (v_right_7) at (4.0cm,-1.8cm) {8};
  \node[bipartitenode, fill=black, text=white] (v_right_8) at (4.0cm,-2.7cm) {9};
  \node[bipartitenode, fill=black, text=white] (v_right_9) at (4.0cm,-3.6cm) {10};
  
  \draw[line width=3.4] (u_left_0) -- (v_right_5);
  \draw[edge, color=\cola] (u_left_0) -- (v_right_5);
  
   \draw[line width=3.4] (u_left_0) -- (v_right_6);
  \draw[edge, color=\cola] (u_left_0) -- (v_right_6);

   \draw[line width=3.4] (u_left_0) -- (v_right_7);
  \draw[edge, color=\cola] (u_left_0) -- (v_right_7);

   \draw[line width=3.4] (u_left_0) -- (v_right_9);
  \draw[edge, color=\cola] (u_left_0) -- (v_right_9);
   \draw[line width=3.4] (u_left_2) -- (v_right_5);
  \draw[edge, color=\cola] (u_left_2) -- (v_right_5);

   \draw[line width=3.4] (u_left_2) -- (v_right_6);
  \draw[edge, color=\cola] (u_left_2) -- (v_right_6);
  
   \draw[line width=3.4] (u_left_2) -- (v_right_7);
  \draw[edge, color=\cola] (u_left_2) -- (v_right_7);
  
   \draw[line width=3.4] (u_left_2) -- (v_right_9);
  \draw[edge, color=\cola] (u_left_2) -- (v_right_9);

  \draw[line width=3.4] (u_left_3) -- (v_right_5);
  \draw[edge, color=\colc] (u_left_3) -- (v_right_5);

  \draw[line width=3.4] (u_left_3) -- (v_right_6);
  \draw[edge, color=\colc] (u_left_3) -- (v_right_6);

  \draw[line width=3.4] (u_left_3) -- (v_right_7);
  \draw[edge, color=\colc] (u_left_3) -- (v_right_7);

  \draw[line width=3.4] (u_left_4) -- (v_right_5);
  \draw[edge, color=\colc] (u_left_4) -- (v_right_5);

  \draw[line width=3.4] (u_left_4) -- (v_right_6);
  \draw[edge, color=\colc] (u_left_4) -- (v_right_6);

  \draw[line width=3.4] (u_left_4) -- (v_right_7);
  \draw[edge, color=\colc] (u_left_4) -- (v_right_7);

  \draw[line width=3.4] (u_left_0) -- (v_right_8);
  \draw[edge, color=lime] (u_left_0) -- (v_right_8);

   \draw[line width=3.4] (u_left_1) -- (v_right_8);
  \draw[edge, color=lime] (u_left_1) -- (v_right_8);

   \draw[line width=3.4] (u_left_4) -- (v_right_9);
  \draw[edge, color=ACMPurple] (u_left_4) -- (v_right_9);

   \draw[color=\cola, draw=black, line width=1.4pt]
    ($(5.7,-0.6-0.2)-(0,.2)$) rectangle ($(10.4,0.1-0.2)+(0,.2)$);

   \draw[color=\cola, draw=\cola, line width=1.0pt]
    ($(5.7,-0.6-0.2)-(0,.2)$) rectangle ($(10.4,0.1-0.2)+(0,.2)$);

  \draw[color=\colc, draw=black, line width=1.4pt]
    ($(5.7,-2.3+\uphigh)-(0,.2)$) rectangle ($(10.4,-1.7+\uphigh)+(0,.2)$);
    \draw[color=\colc, draw=\colc, line width=1.0pt]
    ($(5.7,-2.3+\uphigh)-(0,.2)$) rectangle ($(10.4,-1.7+\uphigh)+(0,.2)$);

  \draw[color=lime, draw=black, line width=1.4pt]
    ($(5.7,-3.1+\uphigh)-(0,.2)$) rectangle ($(10.4,-2.9+\uphigh)+(0,.2)$);
    \draw[color=lime, draw=lime, line width=1.0pt]
    ($(5.7,-3.1+\uphigh)-(0,.2)$) rectangle ($(10.4,-2.9+\uphigh)+(0,.2)$);

  \draw[color=ACMPurple, draw=black, line width=1.4pt, dashed]
    ($(5.7,-3.9+\uphigh)-(0,.2)$) rectangle ($(8.7,-3.7+\uphigh)+(0,.2)$);
    \draw[color=ACMPurple, draw=ACMPurple, line width=1.0pt, dashed]
    ($(5.7,-3.9+\uphigh)-(0,.2)$) rectangle ($(8.7,-3.7+\uphigh)+(0,.2)$);

      \draw[color=ACMPurple, draw=black, line width=1.4pt]
    ($(8.85,-3.9+\uphigh)-(0,.2)$) rectangle ($(10.4,-3.7+\uphigh)+(0,.2)$);
    \draw[color=ACMPurple, draw=ACMPurple, line width=1.0pt]
    ($(8.85,-3.9+\uphigh)-(0,.2)$) rectangle ($(10.4,-3.7+\uphigh)+(0,.2)$);

  \node[clausenode, text=black] (clause_u_left_0) at (6.5, 0.0-0.2) {$\overline{x_{1}} \lor \ell_1$};
  \node[clausenode, text=black] (clause_u_left_2) at (8.0, 0.0-0.2) {$\overline{x_{3}} \lor \ell_1$};
  \node[clausenode, text=black] (clause_v_right_5) at (6.5, -0.4-0.2) {$\overline{x_{7}} \lor \overline{\ell_1}$};
  \node[clausenode, text=black] (clause_v_right_6) at (8.0, -0.4-0.2) {$\overline{x_{8}} \lor \overline{\ell_1}$};
  \node[clausenode, text=black] (clause_v_right_7) at (9.5, 0.0-0.2) {$\overline{x_{6}} \lor \overline{\ell_1}$};
  \node[clausenode, text=black] (clause_v_right_9) at (9.5, -0.4-0.2) {$\overline{x_{10}} \lor \overline{\ell_1}$};
  \node[clausenode, text=black] (clause_u_left_3) at (6.5, -1.8 +\uphigh) {$\overline{x_{4}} \lor \ell_2$};
  \node[clausenode, text=black] (clause_u_left_4) at (6.5, -2.2+\uphigh) {$\overline{x_{5}} \lor \ell_2$};
  \node[clausenode, text=black] (clause_v_right_5) at (8.0, -1.8+\uphigh) {$\overline{x_{6}} \lor \overline{\ell_2}$};
  \node[clausenode, text=black] (clause_v_right_6) at (8.0, -2.2+\uphigh) {$\overline{x_{7}} \lor \overline{\ell_2}$};
  \node[clausenode, text=black] (clause_v_right_7) at (9.5, -1.8+\uphigh) {$\overline{x_{8}} \lor \overline{\ell_2}$};
  \node[clausenode, text=black] (clause_u_left_0) at (6.5, -3.0+\uphigh) {$\overline{x_{1}} \lor \ell_3$};
    \node[clausenode, text=black] (clause_v_right_8) at (8.0, -3.0+\uphigh) {$\overline{x_{2}} \lor \ell_3$};
  \node[clausenode, text=black] (clause_v_right_8) at (9.5, -3.0+\uphigh) {$\overline{x_{9}} \lor \overline{\ell_3}$};
  \node[clausenode, text=black] (clause_u_left_1) at (6.5, -3.8000000000000003+\uphigh) {$\overline{x_{5}} \lor \ell_4$};
  \node[clausenode, text=black] (clause_v_right_8) at (8.0, -3.8000000000000003+\uphigh) {$\overline{x_{10}} \lor \overline{\ell_4}$};
  \node[clausenode, text=black] (eclause_eu_left_1) at (9.6, -3.8+\uphigh) {$\overline{x_{5}} \lor \overline{x_{10}}$};








\end{tikzpicture}
